\documentclass[acmtog,nonacm]{acmart}
\AtBeginDocument{%
  }

\setcopyright{none}
\acmISBN{978-1-4503-XXXX-X/2018/06}

\usepackage[ruled,vlined]{algorithm2e}
\usepackage{float}
\usepackage{subcaption}
\newtheorem{remark}{Remark}
\begin{document}

\title{From Walks to Skeletons: A Reusable Transport Architecture for Interactive Monte Carlo PDE Solvers}

\author{Yunfeng Liao}
\authornote{Both authors contributed equally to this research.}
\author{Yucheng Wu}
\authornotemark[1]

\author{Xiucheng Li}
\affiliation{\institution{Harbin Institute of Technology (Shenzhen)}\city{Shenzhen}\state{Guangdong}\country{China}}









\begin{abstract}
Walk-on-Spheres (WoS) enables mesh-free solutions of elliptic PDEs, but conventional methods repeatedly discard sampled trajectories, limiting their efficiency in interactive applications. We introduce reusable transport skeletons, a persistent representation that organizes stochastic walks into a shared computational graph. Our framework enables efficient reuse across spatial queries and boundary conditions, while supporting progressive refinement, structure-preserving compression, and local adaptation to dynamic geometry. It further enables interactive editing of transport between selected regions. Experiments demonstrate substantial computational savings while preserving solution accuracy. By transforming transient stochastic trajectories into reusable transport structures, our method provides a flexible foundation for interactive Monte Carlo PDE solving.
\end{abstract}



\begin{teaserfigure}
  \includegraphics[width=\textwidth]{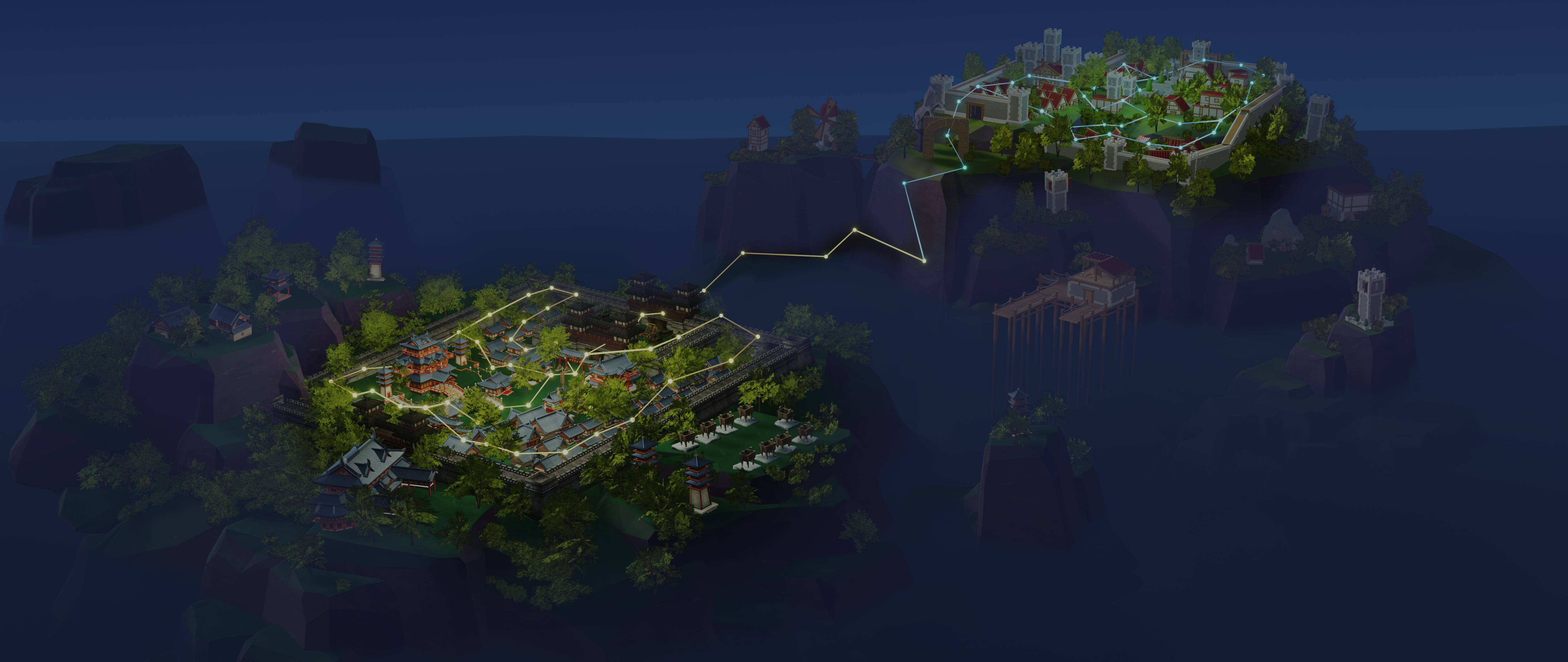}
\caption{\textbf{Interactive PDE solving with reusable stochastic transport.} Illuminated paths visualize transport connections across a complex multi-island environment. Rather than independently resampling trajectories, a persistent and shared transport skeleton retains sampled continuations, adaptive to new queries by reusing existing paths, boundary-condition changes by graph reevaluation, and local geometry edits by preserving unaffected transport while repairing invalid connections.}
  \Description{Enjoying the baseball game from the third-base
  seats. Ichiro Suzuki preparing to bat.}
  \label{fig:teaser}
\end{teaserfigure}


\maketitle
\thispagestyle{fancy}\fancyfoot[C]
{\footnotesize Preprint}

\section{Introduction}
\label{sec:introduction}

Interactive PDE-based applications often revisit the same domain while changing boundary values, adding queries, moving objects, or adjusting how different regions influence one another. Walk on Spheres (WoS) is attractive in this setting because it estimates solutions of elliptic PDEs directly at selected points, without constructing a volumetric mesh \citep{sawhney2020monte}. Yet its usual computational model is transient: each query launches stochastic walks, evaluates their boundary or source contributions, and discards the trajectories. As a result, even small changes can trigger repeated exploration of largely unchanged geometry.

Prior work has reduced this redundancy by caching solution values, derivatives, or local harmonic approximations \citep{miller2023boundary,bakbouk2023mean,zhou2025harmonic}. More recent methods reuse neighboring trajectories, fixed spherical probes, or sampled first-passage operators \citep{czekanski2026neighbors,huang2026probes,jambon2026subdomains}. WoS-based cage cutting also preserves paths unaffected by local cuts \citep{saeed2026cutting}. These methods establish that stochastic information need not be discarded, but address reuse through particular query neighborhoods, spatial interfaces, or applications. What remains useful for interactive workflows is a representation that retains the \emph{dependencies between sampled continuations} and can itself be queried, refined, compressed, and repaired as the problem changes.

We introduce \emph{reusable transport skeletons}: directed acyclic graphs (DAGs) assembled from WoS paths, in which spatial states reference already completed downstream computations. A new walk can attach to this shared structure without repeating its continuation. This connection is not a geometric shortcut: simply snapping nearby states together introduces bias, while forcing walks into tiny attachment neighborhoods can inflate both graph size and importance-weight variance. We instead turn to Poisson-kernel reweighting, accepting attachments only under a prescribed bound on their importance factors. Boundary values and sampled source terms then enter as leaves of the same graph, allowing its retained transport to be evaluated by a dependency traversal when the data change.

Making such a graph reusable also requires preserving spatial access as it evolves. We introduce conditional path refinement to create valid intermediate continuation states from retained paths, and shared spherical interfaces to collect reusable continuations from multiple walks. Progressive and variance-guided sampling extends the graph where additional information is most valuable, while coverage-preserving compression removes redundant states without sacrificing access for future queries. For geometry edits, we retain unaffected transport, repair invalidated flows locally via Doob's $h$-transform to maintain sampled downstream trajectories instead of casting them away like~\citep{saeed2026cutting}, and use bidirectional WoS~\citep{qi2022bidirectional} to capture tiny objects significantly affecting the solution. The resulting representation additionally exposes regional transport interactions that can be modified through on-demand low-rank corrections, without rebuilding the underlying graph.
Our contributions are threefold:
\begin{itemize}
\item A persistent stochastic transport DAG enables reevaluation under changing boundary and source data through bounded-weight path attachment.
\item Conditional path refinement and shared hub interfaces extend spatial reuse, while variance-guided refinement and graph compression balance sampling effort and storage.
\item Interactive transport adaptation repairs affected transitions without discarding shared continuations, with matrix-free low-rank edits enabling regional transport control.
\end{itemize}

\section{Related Work}
\label{sec:related-work}

\paragraph{Caching and Transport Reuse for Grid-Free Monte Carlo PDE Solvers.}
WoS \citep{muller1956continuous}, introduced to geometry processing by \citet{sawhney2020monte}, solves elliptic PDEs pointwise without volumetric meshes. Extensions address varying coefficients \citep{sawhney2022gridfree}, mixed Neumann--Dirichlet conditions via Walk on Stars \citep{sawhney2023stars}, and reverse Green transport via bidirectional WoS \citep{qi2022bidirectional}. Boundary Value Caching stores values and derivatives \citep{miller2023boundary}; Mean Value Caching shares volume-mean estimates \citep{bakbouk2023mean}; Harmonic Caching combines local expansions \citep{zhou2025harmonic}; and Talking to Neighbors reuses walks for nearby queries \citep{czekanski2026neighbors}. Probe-based WoS organizes intermediate states on fixed spheres \citep{huang2026probes}, while Walk on Decomposed Subdomains (WoDS) couples reusable first-passage operators via sparse deterministic systems \citep{jambon2026subdomains}. We likewise share stochastic computation, but retain continuation dependencies in a dynamically built DAG supporting weighted attachment, progressive refinement, and coverage-preserving compression without fixed probes or discretized interfaces.

\paragraph{Adaptive Sampling and Interactive Updates.}
Learned directional and variance-aware path guiding in rendering \citep{muller2017guiding,rath2020variance} motivates steering WoS, despite unstable likelihood ratios. Differential WoS estimates geometry/boundary sensitivities \citep{miller2024differential}; WoDS updates affected local operators \citep{jambon2026subdomains}; and WoS-based cage cutting reuses cut-unaffected paths to update harmonic weights \citep{saeed2026cutting}. Beyond per-query guiding or cutting-specific reuse, our bounded-weight attachment and variance-guided refinement maintain transport for repeated boundary solves, local geometry repair, and low-rank regional edits.

\section{Reusable WoS Skeleton Construction}
Let $\Omega\subset\mathbb R^3$ be the domain with Dirichlet boundary data $g$.
For the source-free problem, the solution is represented by
$u(x)=\mathbb{E}[g(X_\tau)\mid X_0=x]$, where $X_\tau$ is the boundary
arrival of Brownian motion started at $x$, and $\tau$ its first exit time, allowing WoS to estimate this expectation via sampled walks. Independent queries repeatedly explore similar regions and recompute their continuations. The Markov property motivates retaining these continuations since once a spatial state is fixed, its remaining evolution does not depend on the prefix that reaches it. We therefore retain sampled WoS trajectories as directed chains, with multiple trajectories forming outgoing branches. When a new walk reaches a state $p$, it may reuse the completed continuations of an existing node $q$ by connecting to its downstream branches rather than generating new walks termed \textit{path attachment}. Such shared dependencies form a transport directed acyclic graph (DAG), dubbed \textit{skeleton}. 

\textit{Skeleton Data-structure}. Each DAG node $p$ stores a position $x_p$ consisting of these types: query nodes are the locations where the solution is requested, boundary nodes supply terminal data, and the remaining states are internal nodes that a certain WoS path passes by. An edge $p\to q$ means that the estimate at $p$ depends on the continuation at $q$ and let $\mathcal C(p)$
denote these outgoing dependencies. To support reuse, the skeleton also retains how each outgoing branch was sampled. For a completed node $i$, let $\mathcal S(i)$ denote its sampling groups. Each group $s$ records the original sphere $B_s=B(c_s,R_s)\subset\Omega$, its surface-area proposal density $r_s$, and the downstream nodes $\mathcal C_s(i)$ sampled on $\partial B_s$, so that $\mathcal C(i)=\bigcup_{s\in\mathcal S(i)}\mathcal C_s(i)$. A node may contain groups from different spheres because attachments inherit the sampling records of reused branches. These records allow the same sampled continuations to be reweighted for new starting positions. For a fixed skeleton, internal nodes $\hat u_p$ and boundary nodes $\hat u_b$  are given by 
\[
\hat u_p=\sum_{q\in\mathcal C(p)}w_{pq}\hat u_q,
\qquad
\hat u_b=g(x_b),
\]
where $b$ denotes a boundary node and edge weight $w_{pq}$ includes averaging and
importance corrections. Boundary contributions propagate toward queries
in reverse topological order, evaluating each shared continuation once. Since the graph retains boundary positions, changing boundary data supplies new leaf
values without retracing the walks, enabling us to estimate queried values via a cheap traversal on the DAG. 



The main challenge of path attachment is that the current state and the reused node generally occupy different positions. Directly inheriting the latter's sampled continuations introduces bias, while restricting attachment to nearly identical states limits opportunities for reuse. More aggressive attachment like directional guiding may improve reuse but may produce unstable importance weights as further analyzed in Appendix~\ref{app:guided-attachment}. We first introduce a change-of-measure attachment rule with bounded importance weights in Section~\ref{sec:path-attachment} motivated by~\citep{czekanski2026neighbors} to address this tradeoff, and then describe efficient evaluation of the resulting DAG in Section~\ref{sec:graph-evaluation}, followed by progressive refinement of uncertain continuations for skeleton nodes with high variances in Section~\ref{sec:progressive-refinement}. Finally, Section~\ref{sec:source-dag} extends the skeleton construction to Poisson equations with source terms.

\subsection{Path Attachment}
\label{sec:path-attachment}
Path attachment enables a walk at $p$ to terminate early by reusing
the completed continuations of an existing skeleton node $q$.
Rather than treating $p$ as $q$, we connect it to the downstream
branches $\mathcal C(q)$, avoiding new walks along these continuations.
This raises three questions: how to reuse samples generated from a
different starting position, when such reuse is sufficiently stable,
and how to reach reusable regions efficiently.

\textit{Reweighting.}
The sampled continuations of $q$ cannot be directly reused at $p$,
since their Brownian exit distributions differ. We account for this
difference through a change of measure on each recorded sampling
sphere. Attachment requires $p\in B_s$ for every $s\in\mathcal S(q)$.
For a branch $j\in\mathcal C_s(q)$, its weight becomes
\begin{equation}
    w_{pj}
    =
    \frac{K_{B_s}(p,x_j)}{K_{B_s}(q,x_j)}w_{qj},
    \label{eq:skeleton-attachment}
\end{equation}
where $K_{B_s}(x,z)$ is the exit density with respect to surface area
on $\partial B_s$ for a walk started at $x$, and $x_j$ is the
position of node $j$. The existing weight $w_{qj}$ includes the
averaging factor and original proposal correction, so the complete
importance factor of the new connection is
\begin{equation}
    \ell_{p,s}(z)
    =
    \frac{K_{B_s}(p,z)}{K_{B_s}(q,z)}
    \frac{K_{B_s}(q,z)}{r_s(z)}
    =
    \frac{K_{B_s}(p,z)}{r_s(z)}.
\end{equation}

\textit{Admissibility and candidate search.} Although reweighting corrects the exit distribution, large importance factors can increase estimator variance. Moreover, checking every completed node for attachment would undermine the efficiency gained from reuse. We address both issues by defining admissible regions that can be spatially indexed. A connection to $q$ is admissible only if $p$ lies inside every recorded sampling sphere and
\begin{equation}
    \Gamma_q(p)
    :=
    \max_{s\in\mathcal S(q)}
    \sup_{z\in\partial B_s}
    \ell_{p,s}(z)
    \le \Lambda_{\mathrm{attach}},
    \label{eq:attachment-admissibility}
\end{equation}
where $\Lambda_{\mathrm{attach}}$ is the prescribed attachment limit. These conditions define a reusable region for each completed node. We spatially index these regions and select the admissible candidate with the smallest $\Gamma_q(p)$. Upon acceptance, $p$ is connected to $\mathcal C(q)$ using Eq.~\ref{eq:skeleton-attachment}, and its walk terminates. If no admissible candidate exists, the walk continues. Since all reused subgraphs are already completed, these connections preserve acyclicity.

\textit{Directional guiding.} Even with efficient candidate search, uniform WoS walks may rarely enter reusable regions, especially those that are difficult to reach or sparsely covered by the skeleton. Guiding steers exploration toward these hard-to-reach regions in the form of a prescribed point or a selected object surface. Specifically, we sample directions from a mixture of the uniform distribution and a von Mises--Fisher distribution directed toward the target on each WoS sphere. In implementation, the step weight is multiplied by $(4\pi Q(\omega))^{-1}$ for the resulting density of the unit direction $\omega$.

\subsection{Graph Evaluation}
\label{sec:graph-evaluation}



Once constructed, the skeleton can be evaluated for new boundary data without tracing additional WoS paths. Since each edge $i\to j$ represents a dependency on the continuation at $j$, terminal values $\hat u_b=g(x_b)$ propagate toward queries in reverse topological order. This follows a standard parallel DAG traversal, where nodes are evaluated once all their dependencies are resolved, as summarized in Algorithm~\ref{alg:parallel-dag-eval}. Each node is evaluated once and each edge contributes one weighted accumulation, yielding $O(|V|+|E|)$ work.

\begin{algorithm}[t]
\caption{Parallel Evaluation of a WoS Skeleton}
\label{alg:parallel-dag-eval}
\KwIn{DAG $\mathcal G=(V,E)$, weights $w_{ij}$, terminal set $\mathcal T$ with values $v_t$}
\KwOut{Values $\hat u_i$ at query nodes}

Initialize $\mathrm{acc}[i]\gets0$ and
$\mathrm{pending}[i]\gets|\mathcal C(i)|$ for all $i\in V$\;
Set $\hat u_t\gets v_t$ for all $t\in\mathcal T$ and
$\mathcal Q\gets\mathcal T$\;

\While{$\mathcal Q$ is not empty}{
    $\mathcal Q_{\mathrm{next}}\gets\varnothing$\;
    \ForEach{$j\in\mathcal Q$ in parallel}{
        \ForEach{incoming edge $i\to j$ in parallel}{
            $\mathrm{atomicAdd}(\mathrm{acc}[i],\,w_{ij}\hat u_j)$\;
            \If{$\mathrm{atomicSub}(\mathrm{pending}[i],1)=1$}{
                $\hat u_i\gets\mathrm{acc}[i]$\;
                insert $i$ into $\mathcal Q_{\mathrm{next}}$ atomically\;
            }
        }
    }
    $\mathcal Q\gets\mathcal Q_{\mathrm{next}}$\;
}
\end{algorithm}

\paragraph{Topology-based scheduling.}
To reduce the scheduling overhead of frontier-based evaluation, particularly along single-dependency chains, we exploit the fixed skeleton topology to precompute an execution schedule. One scheduling variant fuses maximal single-dependency runs and assigns levels only to the remaining scheduled vertices. Let $\bar{\mathcal C}(i)$ denote the first scheduled descendants reached from
$i$ after skipping fused vertices. Set $\ell(i)=0$ if
$\bar{\mathcal C}(i)=\varnothing$, and otherwise
$\ell(i)=1+\max_{j\in\bar{\mathcal C}(i)}\ell(j)$.
Increasing levels respect the dependency order. A worker dispatched at a
scheduled vertex evaluates its owned run of single-dependency predecessors
locally, omitting those vertices from separate frontiers. This variant
reduces intermediate queue operations while preserving the weighted DAG.

\subsection{Progressive Refinement}
\label{sec:progressive-refinement}

Beyond reusing existing paths, the skeleton also enables more targeted allocation of additional samples. Rather than repeatedly launching new walks from query points, we can refine selected internal nodes by adding new continuations and propagating their updated estimates through the DAG. This allows additional sampling effort to improve multiple queries sharing the same downstream dependencies, without resampling their existing paths.

Let $N_i=|\mathcal C(i)|$ count the existing
outgoing contributions of node $i$, including inherited attachment contributions. If $b_i$ new paths are
added, let $\hat X_{i,r}$ be the value returned by the $r$th independent
vanilla WoS continuation from $x_i$, run to boundary termination without
attachment. We locally update the dependencies of node $i$ by the following to ensure each sample weighs equally: 
\begin{equation}
    \hat u_i
    \leftarrow 
    \frac{N_i}{N_i+b_i}\hat u_i
    +
    \frac{1}{N_i+b_i}\sum_{r=1}^{b_i}\hat X_{i,r}.
    \label{eq:refinement-update}
\end{equation}

\paragraph{Selecting refinement nodes.}
For a target query $q$, candidates are the internal nodes reachable along
its outgoing dependencies. Their influences
$a_i=\frac{\partial\hat u_q}{\partial\hat u_i}$ measure how changes in a node value
affect the query and are obtained by forward propagation from $q$ toward
the boundary. To identify which nodes would benefit most from additional sampling, we also need their continuation variances. Rather than estimating these independently at every node, we partition the domain into subdomains and share a vanilla WoS estimate of single-continuation variance among nodes within each subdomain.
Let $\widehat\sigma_i^2$ denote the estimate assigned to node $i$. Using the local variance surrogate ${a_i^2\widehat\sigma_i^2}/{N_i}$, we rank
candidates by the predicted reduction from one additional continuation as Eq.~\eqref{eq:marginal-variance-reduction} and retain up to $K$ highest-scoring candidates as a fixed set $\mathcal S$.
\begin{equation}
    \Delta_i
    =
    a_i^2\widehat\sigma_i^2
    \left(
        \frac{1}{N_i}-\frac{1}{N_i+1}
    \right)
    =
    \frac{a_i^2\widehat\sigma_i^2}{N_i(N_i+1)}.
    \label{eq:marginal-variance-reduction}
\end{equation}

\paragraph{Budget allocation.}
The preceding selection identifies promising refinement nodes, but does not determine how many additional continuations each should receive. Given a total budget $C$, we allocate $b_i$ continuations to each selected node $i\in\mathcal S$ to minimize the estimated residual variance. Extending the local variance surrogate used for candidate selection gives
\begin{equation}
    \min_{\{b_i\ge0\}}
    \sum_{i\in\mathcal S}
    \frac{a_i^2\widehat\sigma_i^2}{N_i+b_i},
    \qquad
    \text{s.t.}\quad
    \sum_{i\in\mathcal S} b_i=C.
    \label{eq:optimal-refinement-allocation}
\end{equation}
The continuous solution is
$b_i=\max(0,t|a_i|\widehat\sigma_i-N_i)$,
where $t$ is chosen to satisfy the budget constraint.
Consequently, we round the allocations to integers while preserving
the total budget $C$.


\begin{algorithm}[t]
\caption{Variance-Guided Skeleton Refinement}
\label{alg:variance-refinement}
\KwIn{Skeleton $\mathcal G$ with current values, query $q$, continuation budget $C$, candidate limit $K$}
\KwOut{Refined skeleton and updated query values}

Compute influences $a_i$ on the internal nodes reachable from $q$\;

Estimate continuation variances $\widehat\sigma_i^2$ using subdomain partitioning\;

Select $\mathcal S\gets\operatorname{TopK}_i\Delta_i$ using
Eq.~\eqref{eq:marginal-variance-reduction}\;

Allocate integer additions $b_i$ by
Eq.~\eqref{eq:optimal-refinement-allocation} and budget-preserving rounding\;

Add $b_i$ vanilla WoS continuations at each selected node,
rescaling existing weights by Eq.~\eqref{eq:refinement-update}\;

Evaluate new continuations and reevaluate affected ancestors\;
\end{algorithm}

\subsection{Source Terms as an Additive DAG Structure}
\label{sec:source-dag}
For the Poisson equation $\Delta u=f$, a WoS step must account for the
source field $f$ inside its sampling ball as well as the continuation on
the ball's surface. With the Dirichlet Green function defined by
$-\Delta G_B=\delta$ with respect to the Dirac distribution $\delta$, a point
$x\in B$ satisfies
\begin{equation}
    u(x)=\int_{\partial B}K_B(x,z)u(z)\,d\sigma(z)
    -\int_B G_B(x,y)f(y)\,dy.
    \label{eq:poisson-ball-representation}
\end{equation}
The two integrals become separate sets of weighted dependencies in the
skeleton.

\paragraph{Source sampling and graph representation.}
For a newly generated step on $B_s=B(c_s,R_s)$, draw a uniform surface exit
$X$ and $L$ independent volumetric samples $Y_k$ with density
$\rho_s(y)={G_{B_s}(c_s,y)}/{Z_s}$, where
$Z_s=\int_{B_s}G_{B_s}(c_s,y)\,dy={R_s^2} /{6}$ in $\mathbb R^3$. The local estimator is
\begin{equation}
    \widehat u(c_s)=\widehat u(X)
    -\frac{R_s^2}{6L}\sum_{k=1}^{L} f(Y_k),
    \label{eq:centered-source-step}
\end{equation}
where $\widehat u(X)$ is the continuation estimate at the sampled exit.
Thus each step adds one continuation and $L$ source leaves with local
weight $-{R_s^2}/{6L}$. Retaining source locations gives a node $q$ the additive estimator
\begin{equation}
    \hat u_q
    =
    \sum_j w^{K}_{qj}\hat u_j
    +
    \sum_k w^{G}_{qk}f(y_k),
    \label{eq:source-augmented-node}
\end{equation}
where $x_j$ and $y_k$ are its continuation and source locations, while weights
$w^K_{qj}$ and $w^G_{qk}$ include the corresponding sampling and averaging factors. Source leaves supply $f(y_k)$ and do not continue walking or serve as spatial attachment targets.

\paragraph{Attachment and reusable evaluation.}
Each source sample retains its original sampling-group record. To reuse a
group of $q$ from another point $p\in B_s$, keep the sampled locations and
proposal fixed, and replace the Green kernel's starting point through
\begin{equation}
    w^{G}_{pk}
    =
    \frac{G_{B_s}(p,y_k)}
         {G_{B_s}(q,y_k)}
    w^{G}_{qk}.
    \label{eq:source-reuse}
\end{equation}
The continuation edges use the Poisson ratio from
Eq.~\ref{eq:skeleton-attachment}, so attachment reuses both surface and
volume contributions without resampling the target group. And the surface
admissibility condition still follows Section~\ref{sec:path-attachment} governs acceptance.
For fixed geometry, the sampled positions and transport weights are independent of the values of $g$ and $f$. Changing either field therefore requires only new terminal values: boundary leaves receive $g(x_b)$ and source leaves receive $f(y_k)$. Algorithm~\ref{alg:parallel-dag-eval} then evaluates the same DAG without new walks.

\subsection{Spatially Aware Graph Compression}
\label{sec:spatial-compression}

A completed skeleton may contain many intermediate states along WoS trajectories, increasing storage and evaluation cost without necessarily improving reuse. Consecutive single-dependency states can be eliminated exactly, but their removal also discards spatial entry points for future attachment and refinement. We therefore compress the graph while retaining selected states that provide useful spatial coverage or serve as important roles in the skeleton such as queries, terminals, branching and merging vertices. Any remaining chain $S=(v_0,\ldots,v_k)$ whose interior vertices have one incoming and one outgoing edge can be replaced by a single edge with weight $w_S=\prod_{r=0}^{k-1}w_{v_rv_{r+1}}$. This contraction preserves the linear map from terminal values to retained nodes, including all query estimates, while reducing the number of stored vertices and edges. However, eliminated states can no longer serve as aforementioned attachment or following refinement locations. We therefore preserve selected spatial anchors, together with their direct sampled-exit children, associated edges, and sampling records, while contracting the remaining eligible chains.



\paragraph{Coverage-preserving anchor selection.}
An anchor set $\mathcal A$ is selected to provide at least $m\ge1$
attachment choices over a target fraction $\rho_{\mathrm{target}}$ of the
domain. With $\mathcal R_v$ denoting the reusable region of candidate $v$
and $\mathcal P$ a fixed set of uniform probes, coverage is measured by
$c_{\mathcal A}(p)=\sum_{v\in\mathcal A}\mathbf1_{p\in\mathcal R_v}$ and
$\rho_m(\mathcal A)={|\{p\in\mathcal P:c_{\mathcal A}(p)\ge m\}|}/{|\mathcal P|}$,
subject to $\rho_m(\mathcal A)\ge\rho_{\mathrm{target}}$.
Selection follows the greedy criterion
$v^\star=\arg\max_{v\notin\mathcal A}\Delta_{\mathcal A}(v)$, where
$\Delta_{\mathcal A}(v)=\sum_{p\in\mathcal P\cap\mathcal R_v}\mathbf1_{c_{\mathcal A}(p)<m}$
favors candidates covering probes with fewer than $m$ anchors.

Since anchors along one trajectory share downstream samples, coverage can be
extended to distinct trajectory origins to reduce correlation in subsequent
attachment. With $b(v)$ identifying the trajectory that created $v$ and
$\mathcal B_{\mathcal A}(p)=\{b(v):v\in\mathcal A,\ p\in\mathcal R_v\}$
denoting the origins available at $p$, the count becomes
$c_{\mathcal A}(p)=|\mathcal B_{\mathcal A}(p)|$ and the greedy score becomes
$\Delta_{\mathcal A}(v)=\sum_{p\in\mathcal P\cap\mathcal R_v}\mathbf1_{c_{\mathcal A}(p)<m \land b(v)\notin\mathcal B_{\mathcal A}(p)}$.

\section{Adaptive Transport Reuse for Interactive Scenes}

A reusable transport skeleton must support more than repeated evaluation on fixed geometry. In interactive scenes, new queries may appear in poorly covered regions, objects may move or be inserted, and users may wish to edit transport between selected regions. Rebuilding the skeleton after every change would discard much of its reusable computation. We therefore develop local graph updates that adapt the retained transport without reconstructing it globally: path refinement expands spatial attachment opportunities, hub-based reuse shares continuations across dense queries, geometry repair updates transitions affected by scene changes, and regional transport editing provides direct control over selected interactions.

\begin{figure}
    \centering
    \includegraphics[width=1\linewidth]{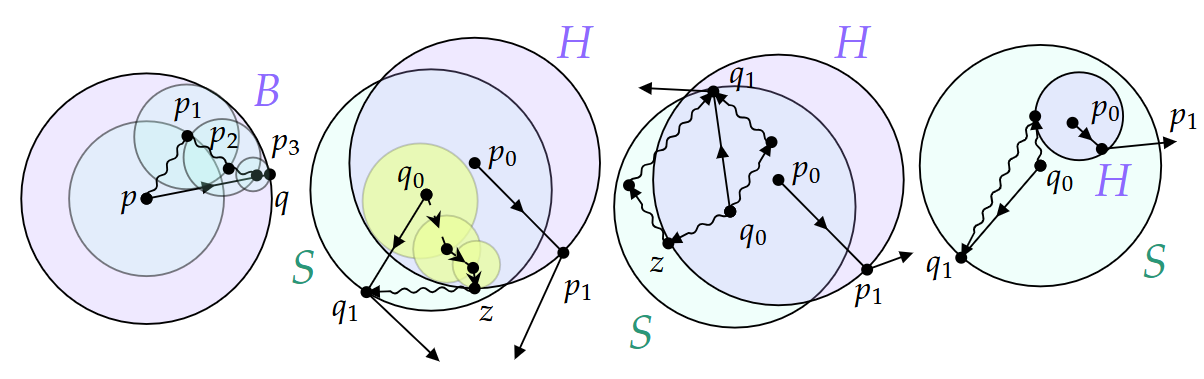}
    \caption{Path refinement and interface construction.}
    \label{fig:path refine and fine-tune}
\end{figure}

\subsection{Path Refinement for New Queries}
\label{sec:path-refinement}

The attachment bound in Section~\ref{sec:path-attachment} limits the region
from which a retained sampling sphere can be reused to avoid variance explosion, while path refinement comes to rescue by exposing intermediate positions within retained WoS steps, providing additional attachment locations for new queries. Path refinement also paves way for both building hub-sphere for intra-skeleton trajectories reuse (Section~\ref{sec:hub-reuse}), and skeleton repairing and tiny-object capturing in dynamic scenarios (Section~\ref{sec:geometry-repair}). 

Consider an ordinary WoS step $p\to q$ with sampling ball $B=B(p,R)$,
where the recorded exit $q$ was sampled uniformly on $\partial B$.
Starting at $p_0=p$, one can generate intermediate positions conditioned on
that exit as Figure~\ref{fig:path refine and fine-tune}. For a local ball $B_k=B(p_k,R_k)\subset B$, the next position
$s\in\partial B_k$ is sampled with surface-area density
\begin{equation}
f_k(s\mid q)
=
\frac{1}{|\partial B_k|}
\frac{K_B(s,q)}{K_B(p_k,q)} .
\label{eq:refinement-conditional-density}
\end{equation}
This Doob $h$-transform, with $h(x)=K_B(x,q)$, recovers Brownian transitions conditioned on exiting $B$ at $q$. We choose $R_k=\bigl(R-\|p_k-p\|\bigr)/2,$ half the remaining distance to the original sphere. This leaves each local ball strictly inside $B$ while advancing the path and exposing intermediate positions. Appendix~\ref{app:refinement-sampler} shows that it requires at most two draws per step on average in rejection sampling for this radius choice.

The original edge $p\to q$ can coexists with the newly generated positions $p_k$, $1\le k\le N-1$, as a node with a unit-weight edge $p_k\to q$. To see why direct connections
suffice, consider a chain $p_1\to\cdots\to p_m\to q$ with unit weights, which makes every node in this chain evaluate to the retained estimate $U_q$ at $q$. In practice, one can replace an ordinary retained edge $p\to q$ of coefficient $w$ by a chain $p\to p_1\to\cdots\to p_m\to q$ uses $w_{p\to p_1}=w/N, w_{p_j\to p_{j+1}}=1,w_{p_m\to q}=1.$ Connecting each node directly to $q$ therefore preserves its value while reducing propagation depth. Notably, though these $p_k$'s can be sampled sequentially, their values can still simply share a direct reference to the same continuation. Such conditional construction ensures that each new node $z$ has the correct continuation law on its own sphere, making it a valid attachment target (Appendix~\ref{app:refinement-proof}).


As a conditioned path approaches $q$, shrinking local balls provide diminishing spatial coverage. We therefore limit each refinement path to $L$ steps or terminate earlier when the sphere clearance falls below a prescribed tolerance, allocating further samples to fresh conditioned paths starting from $p$.

\subsection{Shared Query Evaluation in a Hub}
\label{sec:hub-reuse}

Path attachment and refinement allow individual queries to reuse
existing continuations, but densely distributed queries may still
incur repeated costs in searching for and connecting to the skeleton.
To share these connections, we introduce a hub ball
$H=B(p_0,R_H)\subset\Omega$ containing a group of queries.
Rather than attaching each query separately, we construct a common
pool of interface samples on $\partial H$, each referencing an
existing downstream continuation. Queries inside $H$ then reuse
this pool through the Poisson kernel $K_H$, with only their
query-dependent weights evaluated separately.

These interfaces can be obtained either by conditionally refining retained paths to their crossings with $\partial H$ or by sampling directly on portions of $\partial H$ supported by retained sampling spheres. In the conditional construction (Figure~\ref{fig:path refine and fine-tune}), a retained step $q_0\to q_1$ supported by a ball $S$ is refined within $S\cap H$, conditioned on $q_1$. If the path reaches $\partial H$ before $\partial S$, its first crossing $z\in\partial H$ provides an interface referencing $q_1$ with unit weight, avoiding additional spatial importance ratios. While this construction reuses existing paths without reweighting their reconstructed crossings, direct sampling provides explicit control over the interface distribution, allowing contributions from different sampling groups to be combined across $\partial H$. We therefore develop the direct-sampling estimator below and leave the conditional construction and its probabilistic justification to Appendix~\ref{app:fine-tuning-proof}.


For $L$ retained sampling groups, indexed by $j=1,\ldots,L$, let $S_j$
be the supporting ball of group $j$ and let $U_{Q_j}$ denote the retained
continuation estimate at its exit $Q_j$. Each group contributes over a supported hub patch $\Gamma_j=\partial H\cap S_j$, which covers the entire $\partial H$ if $S_j$ contains $H$. Then independently of its retained exit
$Q_j\sim r_j$, sample $Z\sim g_j(z)={\mathbf1_{\Gamma_j}(z)}/{|\Gamma_j|}$ on $\Gamma_j$,
where $g_j$ is an area density, extended by zero outside the patch. The interface continuation is given by Eq.~\eqref{eq:hub-interface-continuation} where the weight converts the retained proposal density $r_j$ to the Poisson
kernel at $Z$.  
\begin{equation}
\widehat U_j(Z)
=
\frac{K_{S_j}(Z,Q_j)}{r_j(Q_j)}U_{Q_j}.
\label{eq:hub-interface-continuation}
\end{equation}


Since individual sampling groups cover only portions of $\partial H$, their contributions must be combined without double-counting overlapping patches. We therefore sample a group $J=j$ with probability $\beta_j$, followed by an interface $Z\sim g_j$, giving the marginal density $m(z)=\sum_{j=0}^{L}\beta_jg_j(z)$ with $\sum_j\beta_j=1$. As the retained groups may not cover the entire hub surface, we include a full-support fallback group with $g_0(z)=1/|\partial H|$ and $\beta_0>0$, whose continuation $\widehat U_0(Z)$ is generated by vanilla WoS and retained in the skeleton. This ensures $m(z)>0$ throughout $\partial H$, allowing the hub estimator to use $N$ shared interface samples whose unbiasedness is validated by Proposition~\ref{prop:hub-reuse-unbiased}:
\begin{equation}
\widehat u_H(x)
=
\frac1N\sum_{n=1}^{N}
\frac{K_H(x,Z_n)}{m(Z_n)}
\widehat U_{J_n}(Z_n).
\label{eq:hub-mixture-estimator}
\end{equation}

\begin{proposition}[Unbiased hub reuse]
\label{prop:hub-reuse-unbiased}
If the retained downstream estimators are conditionally unbiased and
interface sampling is independent of the retained continuations, then
$\mathbb E[\widehat U_j(Z)\mid Z=z]=u(z)$ and
$\mathbb E[\widehat u_H(x)]=u(x)$ for every $x\in H$.
\end{proposition}

\begin{proof}[Proof sketch]
The ratio in Eq.~\eqref{eq:hub-interface-continuation} recovers the Poisson
representation at $Z$, while the mixture density $m$ cancels the joint
sampling density over $\partial H$.
See Appendix~\ref{app:hub-reuse-proof}.
\end{proof}




To represent the hub estimator in the skeleton, sampling groups and mixture probabilities are chosen independently of the retained exit and continuation samples. Each interface $Z_n$ generated from group $j$ introduces edges with weights
\begin{equation}
w_{x\to Z_n}=\frac{K_H(x,Z_n)}{N\,m(Z_n)},
\qquad
w_{Z_n\to Q_j}=\frac{K_{S_j}(Z_n,Q_j)}{r_j(Q_j)}.
\end{equation}
The first weight accounts for hub sampling and averaging, while the second reweights the retained continuation. Interface positions, mixture densities, and continuation references are stored once per hub, with queries evaluating their own weights against this shared pool. For $M$ queries and $N$ interfaces, this reduces explicit connection storage from $O(MN)$ to $O(M+N)$, although exact evaluation still requires $O(MN)$ kernel contributions.

\subsection{Transport Adaptation for Dynamic Objects}
\label{sec:geometry-repair}

\begin{figure}[t]
    \centering
    \includegraphics[width=1.0\linewidth]{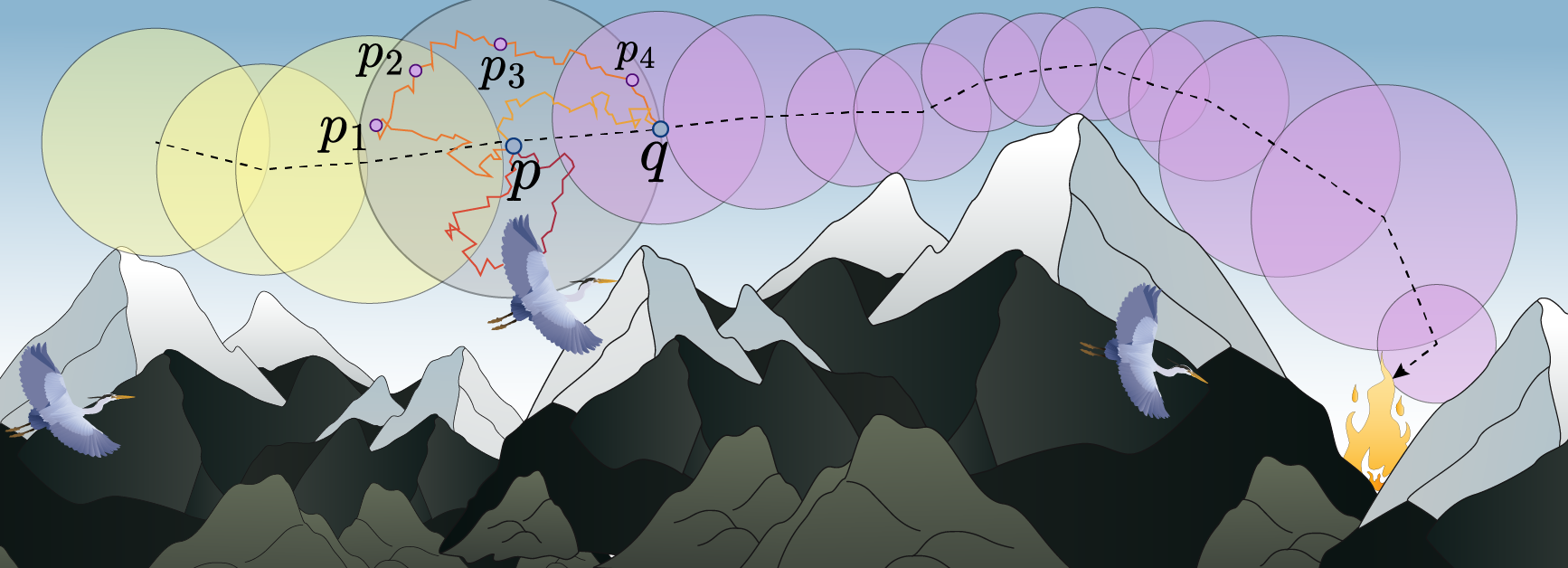}
    \caption{Conditional path refinement and geometry repair. Purple highlights retained continuations, while gray marks a WoS sampling sphere intersected by a moving object. Resampling from scratch in complex scenes may lose previously discovered rare boundary contributions, causing large sampling errors. Instead, conditional refinement preserves existing transport by reconnecting paths that avoid the object to $q$, while redirecting collisions to the new boundary.}
    \label{fig:dynamic_obj}
\end{figure}

As Figure~\ref{fig:dynamic_obj}, dynamic geometry changes can invalidate previously sampled WoS transitions, but their downstream continuations may remain reusable. When a new object intersects a retained sampling ball, the recorded exit alone cannot determine whether the underlying Brownian path would have encountered the object first. Rather than discarding the entire continuation, we use conditional path refinement to reconstruct the transition and redirect it to the new boundary only when necessary. This enables local adaptation of the skeleton while preserving unaffected transport. We consider an object inserted into the domain, with the remaining geometry and boundary data fixed.



We spatially index the sampling balls to identify transitions affected by the inserted object. For an affected edge $p\to q$ supported by a ball $S$, we reconstruct its path conditioned on the recorded exit $q$, using Eq.~\ref{eq:refinement-conditional-density}, with local sampling balls contained in $S$ and outside the inserted object. If the reconstructed path hits the object, the edge is redirected to the new boundary; otherwise, it reconnects to $q$ upon approaching $\partial S$ within a prescribed tolerance, provided $q$ remains outside the object. In both cases, the original edge coefficient is preserved, since the conditional sampling law accounts for the retained exit. To preserve shared transport, repair follows the affected dependencies from queries toward their continuations, replacing only invalidated references, while unaffected downstream subgraphs are retained and only those modified nodes with their ancestors are reevaluated.  This confines both reconstruction and reevaluation to the affected transport.



Conditional repair may suffer from high variance when small inserted objects with significant boundary contributions are rarely hit. Our framework is also compatible with bidirectional WoS (Bi-WoS)~\cite{qi2022bidirectional}, whose reverse Green-function estimates can provide shared local corrections to affected sampling balls, accounting for both retained continuations and new boundary contributions. The integration and stopping criteria are detailed in Appendix~\ref{app:conditional-repair}.

\subsection{Application: Regional Transport Editing}
\label{sec:low-rank transport editing}

The retained skeleton exposes transport dependencies between spatial regions, enabling artist-controlled amplification or attenuation of selected interactions. Rather than rebuilding the solver, we approximate the requested region-to-region transport as a low-rank patch applied to the original DAG.

Eliminating internal states yields an effective transport operator on retained anchors. For retained states $A$ and eliminated states $I$, $\hat{\mathbf u}_I=W_{II}\hat{\mathbf u}_I+W_{IA}\hat{\mathbf u}_A$, giving  $\mathcal T=W_{AA}+W_{AI}(I-W_{II})^{-1}W_{IA}.$  Acyclicity allows $\mathcal T$ and $\mathcal T^{\mathsf T}$ to be applied through graph traversals without explicitly forming either matrix.

For selected source and target clusters $\sigma$ and $\tau$, let $P_\sigma$ inject source values and $R_\tau$ restrict target values, defining $T_{\tau\leftarrow\sigma}=R_\tau\mathcal T P_\sigma$. We approximate this block using randomized range finding with target rank $r$ and $p$ oversampling vectors. Draw $\Omega\in\mathbb R^{|\sigma|\times(r+p)}$ and compute $Y=T_{\tau\leftarrow\sigma}\Omega$, $Q=\operatorname{orth}(Y)$ by forward traversal, followed by $Z=T_{\tau\leftarrow\sigma}^{\mathsf T}Q$ by reverse traversal. This gives $ T_{\tau\leftarrow\sigma}\approx QZ^{\mathsf T}.$ Only $r+p$ batched forward probes and corresponding reverse propagation are needed, and forward probes can be shared across target clusters with the same source.

To scale the selected interaction by $\gamma$, a low-rank correction is applied by $    \hat{\mathbf u}_\tau
    \leftarrow
    \hat{\mathbf u}_\tau
    +(\gamma-1)QZ^{\mathsf T}\hat{\mathbf u}_\sigma.$ The original solver is recovered at $\gamma=1$; other values deliberately modify the represented transport rather than the underlying PDE. However, the source and target clusters may interleave in topological order, so the correction must exclude backward dependencies to preserve the DAG. Let $t(v)$ satisfy $t(i)<t(j)$ for every original edge $i\to j$, and define $C_{ij}=\mathbf1_{t(i)<t(j)}$ for $i\in\tau$, $j\in\sigma$. The causal correction is $\Delta T_{\tau\leftarrow\sigma}=(\gamma-1)C\odot(QZ^{\mathsf T})$, which can be evaluated without forming the mask. During reverse-topological traversal, maintain $\mathbf s_i=\sum_{j\in\sigma,\ t(j)>t(i)}Z_j\hat u_j$, giving $\Delta\hat u_i=(\gamma-1)Q_i^{\mathsf T}\mathbf s_i$ for $i\in\tau$. When crossing source node $j$, update $\mathbf s\leftarrow\mathbf s+Z_j\hat u_j$.

\begin{remark}[Edit Complexity]
Each causal patch costs $O(r(|\sigma|+|\tau|))$ per evaluation, without dense blocks or changes to the DAG. Regional edits can be applied as separate low-rank patches.
\end{remark}

\section{Experiments and Results}

\subsection{Experimental Setup}
\label{sec:experimental-setup}

We evaluate accuracy, solve time, and the amortization of construction
cost, together with refinement, compression, and interactive reuse.
Unless stated otherwise, experiments use the multi-object mesh domain shown in
Figure~\ref{fig:scene-rocks} : irregular obstacles lie inside an enclosing
sphere, with spatially varying scalar Dirichlet data on their surfaces.
Specialized domains are introduced in the corresponding experiments.

\begin{figure}[!htbp]
  \centering
  \includegraphics[width=0.7\linewidth]{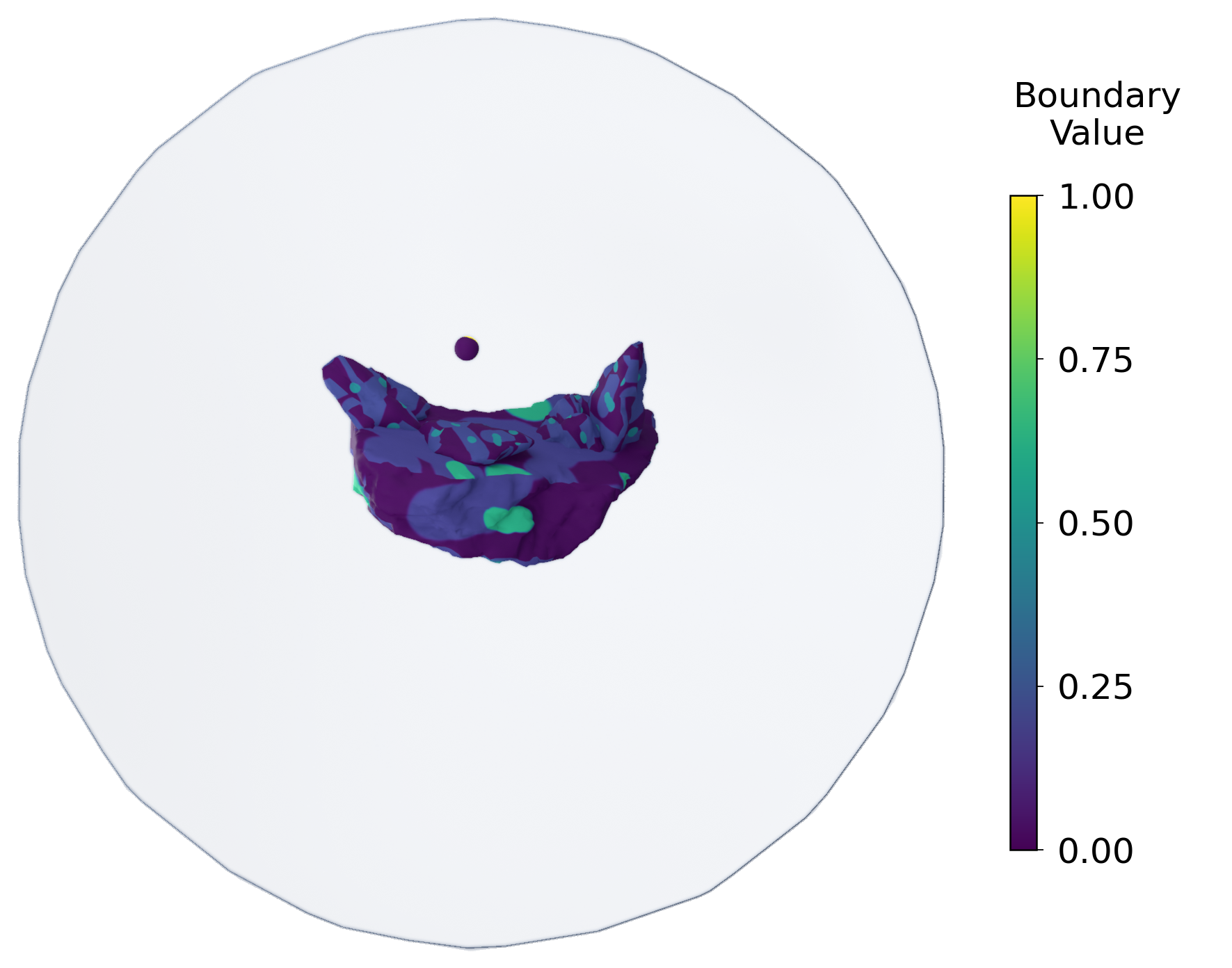}
  \caption{Geometry of the common multi-object test domain, viewed inside a enclosing sphere.}
  \Description{A naturally lit view shows irregular obstacles within a spherical enclosure.}
  \label{fig:scene-rocks}
\end{figure}

A typical construction procedure begins by tracing independent WoS trajectories from a set of initial queries to form an initial skeleton. We then introduce additional queries, whose walks reuse completed continuations through path attachment, producing an attached skeleton. For comparison, an independent-path skeleton traces independent WoS trajectories from all queries without attachment.

Accuracy is measured by
$\mathrm{MSE}=\frac{1}{M}\sum_{i=1}^{M}
(\widehat u(x_i)-u_{\mathrm{ref}}(x_i))^2$,
where $M$ is the total number of queries, $x_i$ is a query position,
$\widehat u$ is the estimated solution, and $u_{\mathrm{ref}}$ is an
independent Monte Carlo reference using $10^5$ WoS
trajectories.  Each configuration was repeated for 5-10 runs, and the mean value was reported. 

All timed methods run serially on the CPU. Skeleton solve time measures
weighted DAG evaluation (Section~\ref{sec:graph-evaluation}).Construction
and solver preparation are included only in the amortization experiment;
dynamic-update timings include repair and reevaluation.  The attachment limit in Eq.~\ref{eq:attachment-admissibility} is set to 4 unless varied explicitly, as recommended by the accuracy--efficiency trade-off observed in Section~\ref{sec:attachment-threshold-experiment}.

\subsection{Skeleton Accuracy and Solve Time}
\label{sec:cache-accuracy-experiment}

Can reuse support a more accurate estimate with a shorter solve time? We compare
an independent-path skeleton, an attached
skeleton that reuses an initial skeleton through path attachment, and
Vanilla WoS. Varying their trajectories per query traces the
accuracy--solve-time curves in Figure~\ref{fig:cache-accuracy}.

\begin{figure}[!htbp]
  \centering
  \captionsetup[subfigure]{font=small,skip=2pt}
  \begin{subfigure}[t]{0.49\linewidth}
    \centering
    \includegraphics[width=\linewidth]{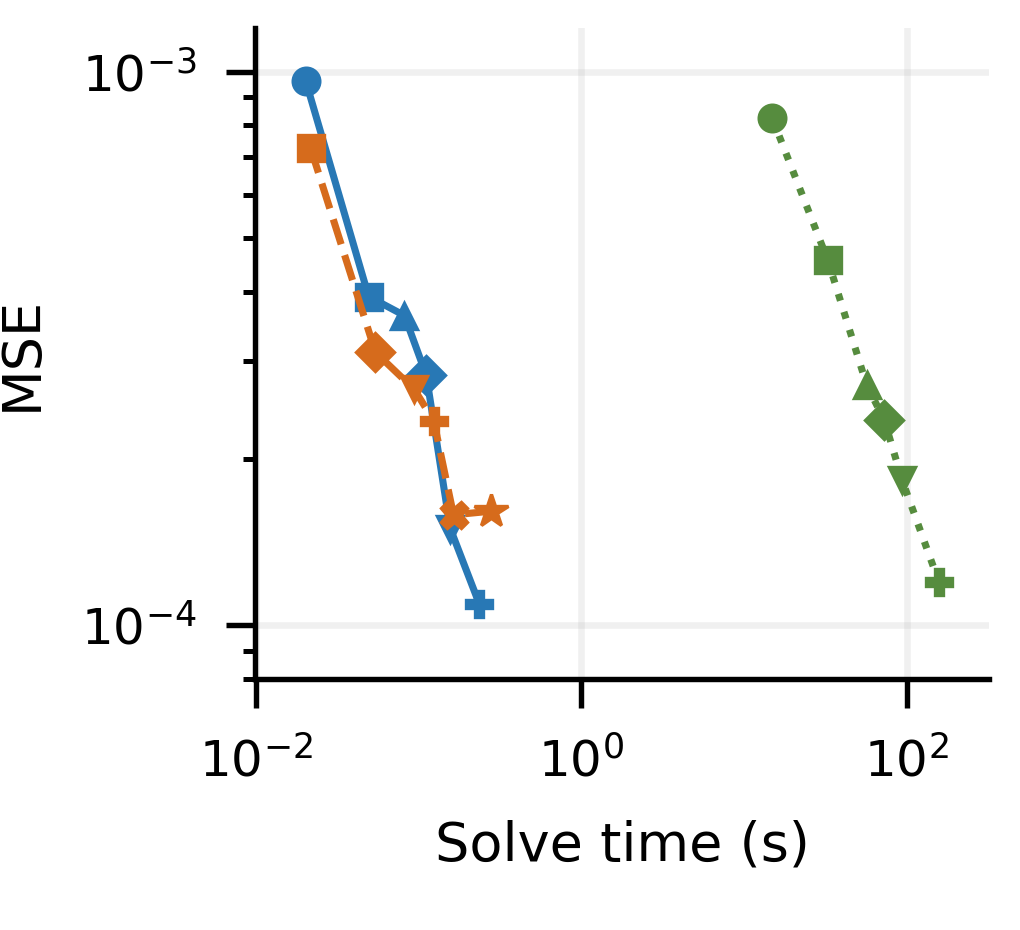}
    \caption{Skeletons and Vanilla WoS}
    \label{fig:cache-accuracy-all}
  \end{subfigure}\hfill
  \begin{subfigure}[t]{0.49\linewidth}
    \centering
    \includegraphics[width=\linewidth]{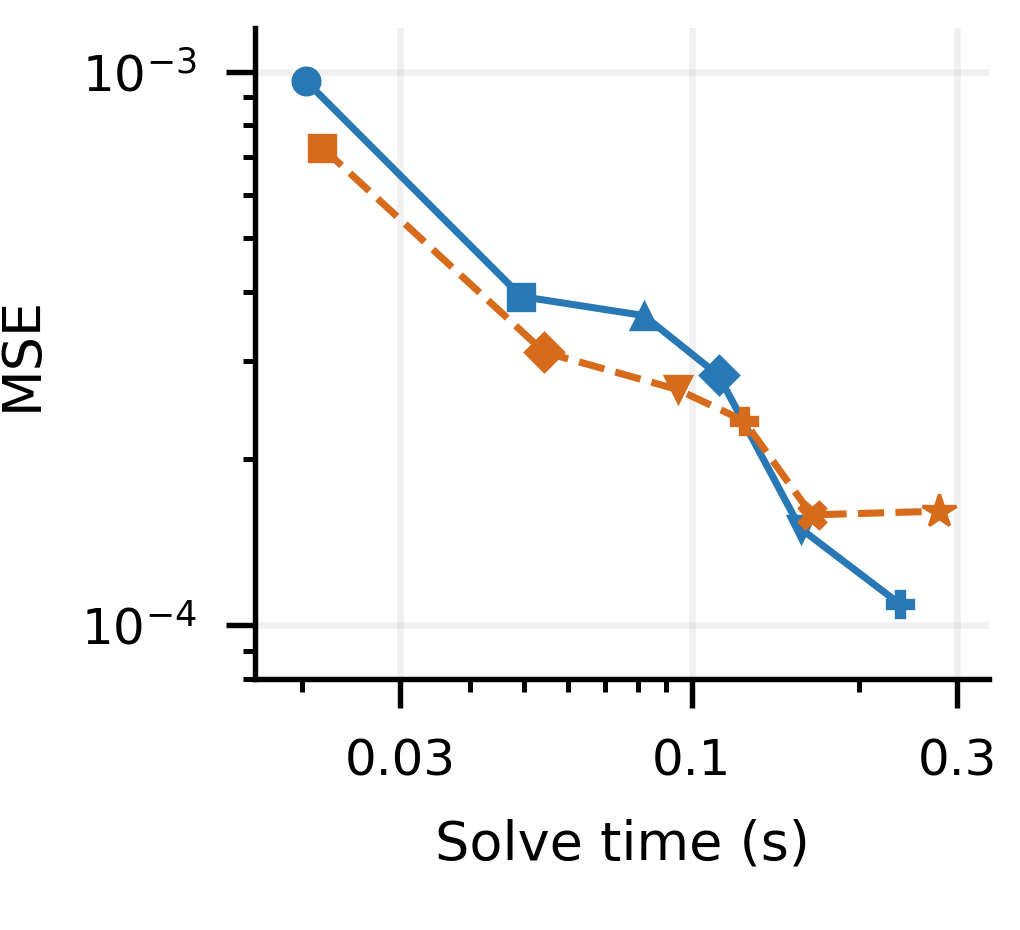}
    \caption{Skeleton detail}
    \label{fig:cache-accuracy-detail}
  \end{subfigure}
  \par\smallskip
  \includegraphics[width=\linewidth]{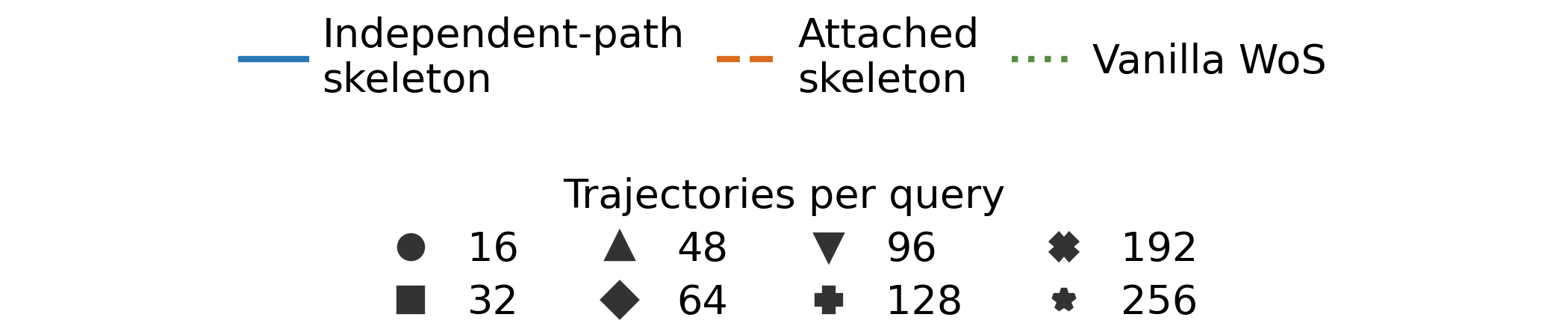}
  \caption{Skeleton evaluation versus Vanilla WoS, averaged over random
  seeds at fixed queries.
  Color identifies the method and marker shape gives the number of
  trajectories per query; for the attached skeleton this budget applies to
  both construction stages. Both axes are logarithmic. Construction and
  solver preparation costs are excluded.}
  \Description{Three accuracy--solve-time curves. Both skeletons solve far
  faster than Vanilla WoS. The attached skeleton initially decreases in error but
  does not retain an accuracy advantage at the largest budgets.}
  \label{fig:cache-accuracy}
\end{figure}

Both skeletons have substantially lower solve time than Vanilla WoS.
For example, the attached skeleton with 64 trajectories per
query reduces MSE by 62\% while solving about $275\times$ faster than
16-trajectory Vanilla WoS. Reusing the skeleton removes repeated geometric
traversal from the solve, allowing a larger construction budget to improve
accuracy at low solve time.

The skeleton detail separates this benefit from attachment's contribution.
At moderate budgets, attachment offers a better accuracy--solve-time trade-off
than the independent-path skeleton: its 96-trajectory configuration has both
lower error and lower solve time than the independent-path 64-trajectory
skeleton. At larger budgets the attached curve flattens while the
independent-path skeleton
continues to improve. This behavior is consistent with shared continuations
limiting the benefit of additional trajectories.

\subsection{Construction-Cost Amortization}
\label{sec:cache-amortization-experiment}

To measure when construction pays for itself, we build an independent-path
skeleton, then solve successively
for different boundary field. Cumulative skeleton time includes construction
and solver preparation once, followed by skeleton evaluation for each field; Vanilla
WoS generates fresh trajectories for each solve. Figure~\ref{fig:cache-amortization}
compares the mean across skeleton builds with one measured Vanilla WoS
solve sequence.

\begin{figure}[!htbp]
  \centering
  \includegraphics[width=\linewidth]{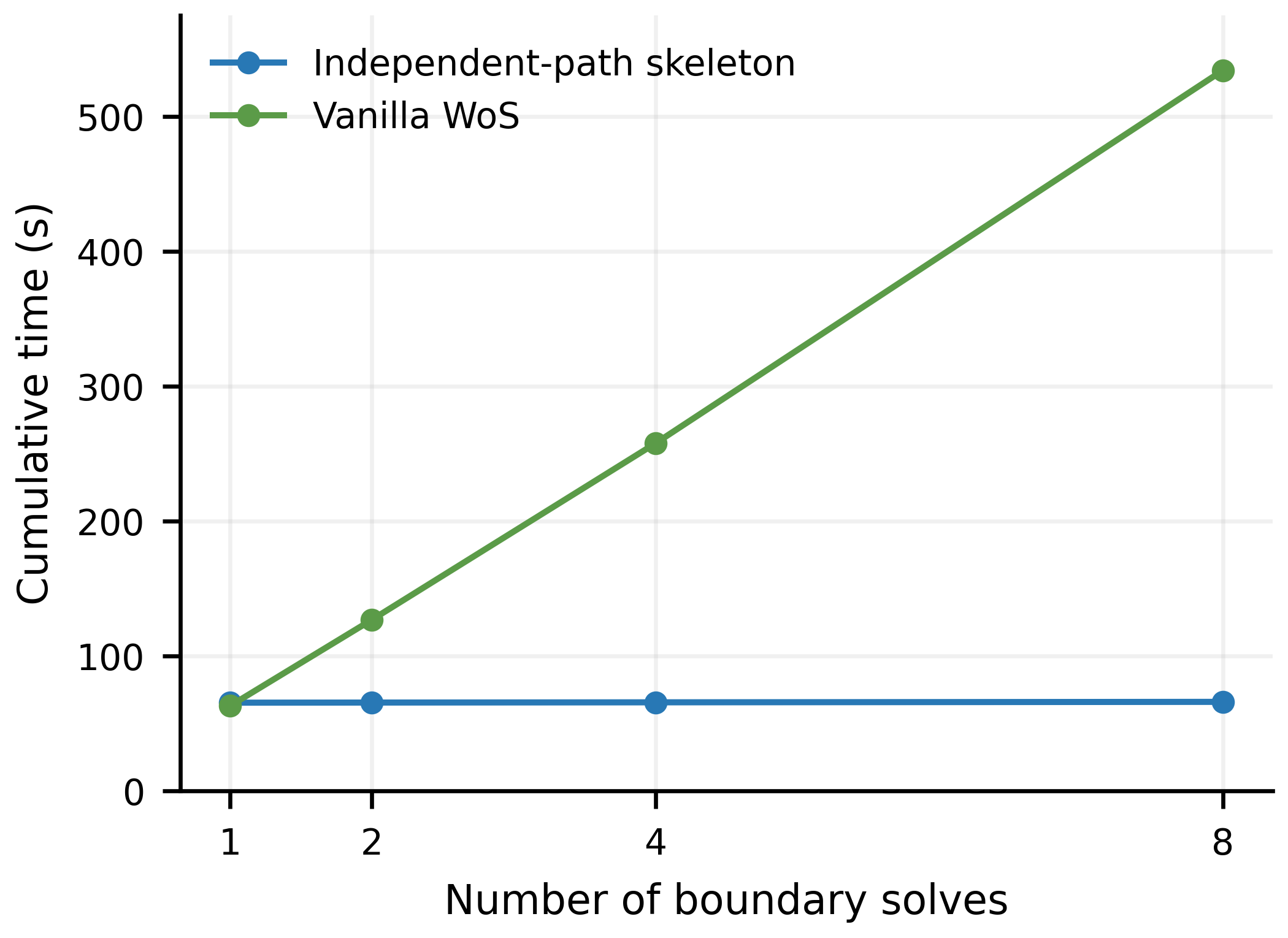}
  \caption{Cumulative cost of repeated boundary solves with 64 trajectories
  per query. Skeleton time includes construction and solver preparation and
  is averaged over random seeds. Vanilla WoS generates new trajectories for each
  boundary field.}
  \Description{Skeleton cumulative time stays near 66 seconds, while Vanilla WoS
  rises from 63 seconds at one solve to 534 seconds at eight solves.}
  \label{fig:cache-amortization}
\end{figure}

Construction costs roughly one Vanilla WoS solve, while subsequent skeleton
evaluations add little to the total. The skeleton amortizes its initial cost
by the second solve and is $8.1\times$ faster cumulatively at eight solves.
Thus, even without attachment, retaining transport separates the cost of
changing boundary data from the cost of exploring the geometry.

\subsection{Sample Budget Allocation}
\label{sec:sample-allocation-experiment}

To investigate how to allocate sampling budget between initial and additional queries to minimize their combined estimation error, 
we divide a fixed trajectory budget between equal-sized batches of initial
and additional queries, varying the fraction allocated to attachment.
Figure~\ref{fig:sample-allocation} reports mean error across matched trials.

\begin{figure}[!htbp]
  \centering
  \includegraphics[width=\linewidth]{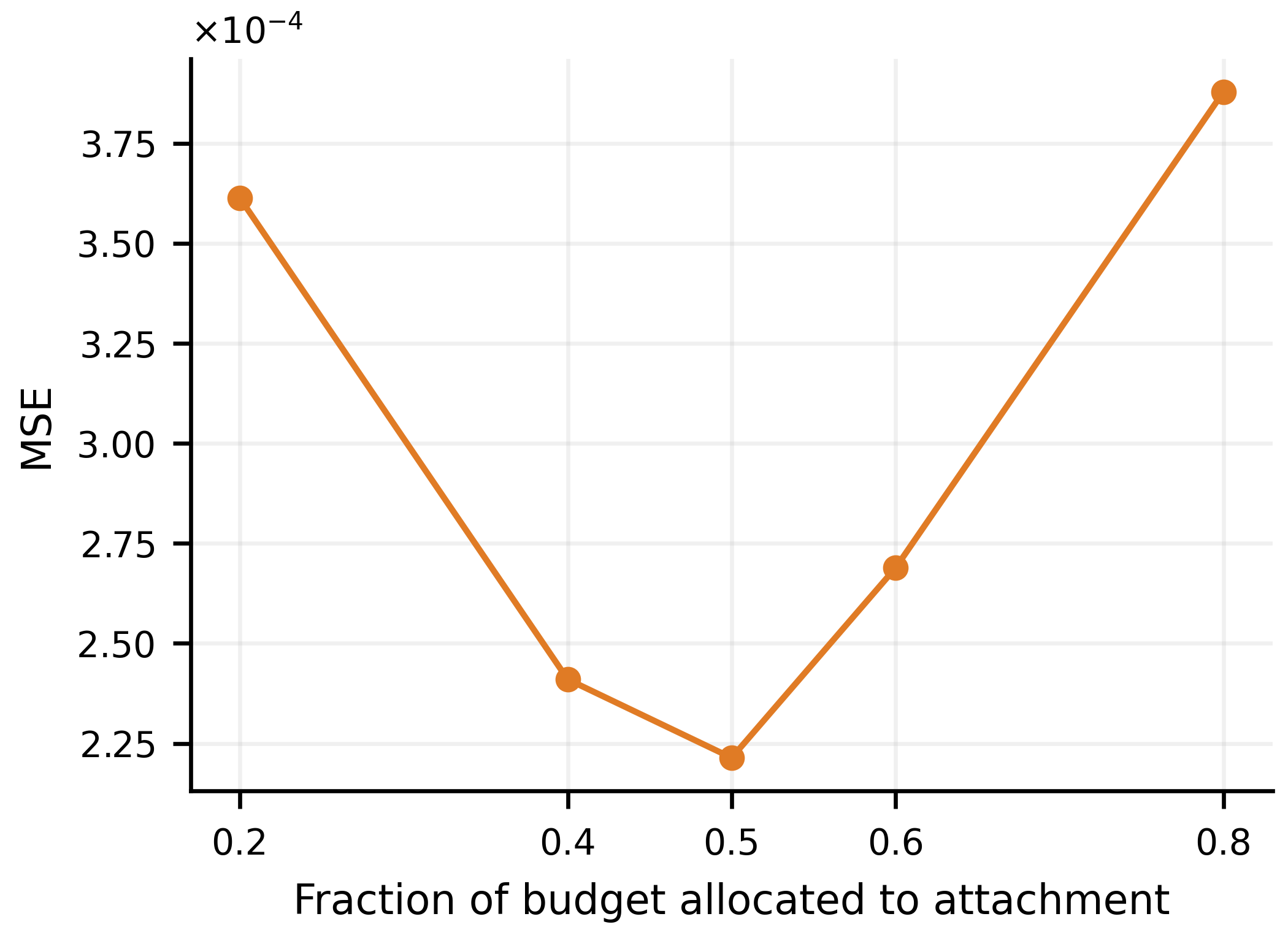}
  \caption{Mean MSE over all initial and additional queries versus the fraction
  of the trajectory budget allocated to attachment. The balanced allocation gives the smallest mean error.}
  \Description{Mean MSE decreases from 3.61 times ten to the minus four at
  fraction 0.2 to 2.21 times ten to the minus four at 0.5, then rises to
  3.88 times ten to the minus four at fraction 0.8.}
  \label{fig:sample-allocation}
\end{figure}

The error curve favors a balanced allocation. Spending too little on initial
construction weakens the transport available for reuse, while spending too
little on attachment undersamples the additional queries. 

\subsection{Attachment Limit Sensitivity}
\label{sec:attachment-threshold-experiment}

In Eq.~\ref{eq:attachment-admissibility}, a smaller attachment limit favors higher accuracy at the cost of reduced attachment efficiency.  To explore this trade-off, we vary this limit while holding the initial skeleton fixed. Table~\ref{tab:attachment-threshold}
reports MSE, the fraction of additional-query walks ending by attachment,
and the mean WoS step count over all additional-query walks.

\begin{table}[!htbp]
  \centering
  \caption{Attachment-limit sensitivity, averaged over random seeds;
  MSE covers both query batches.}
  \label{tab:attachment-threshold}
  \begin{tabular}{rrrr}
    \toprule
    Limit & MSE ($10^{-4}$) & Success (\%) & \shortstack{Mean WoS\\steps} \\
    \midrule
    1    & 2.35 &  0.0 & 34.06 \\
    1.25 & 2.54 & 26.3 & 24.94 \\
    1.5  & 2.91 & 60.0 & 13.52 \\
    2    & 3.03 & 80.9 &  6.55 \\
    4    & 3.62 & 93.0 &  2.77 \\
    8    & 5.76 & 95.8 &  1.99 \\
    16   & 5.42 & 97.1 &  1.67 \\
    \bottomrule
  \end{tabular}
\end{table}

Relaxing the limit initially removes most exploration work, at the cost
of increased error. Beyond 4, attachment is already frequent: further
relaxation buys little additional reuse while allowing larger importance
factors and higher MSE. We therefore use 4 as the default compromise
between accuracy and exploration cost.

\subsection{Poisson Problems with Source Terms}
\label{sec:source-accuracy-experiment}

We next test whether the additive source structure in
Section~\ref{sec:source-dag} retains the accuracy--solve-time benefit for
Poisson problems. Both methods use two source samples per WoS step.

\begin{figure}[!htbp]
  \centering
  \includegraphics[width=\linewidth]{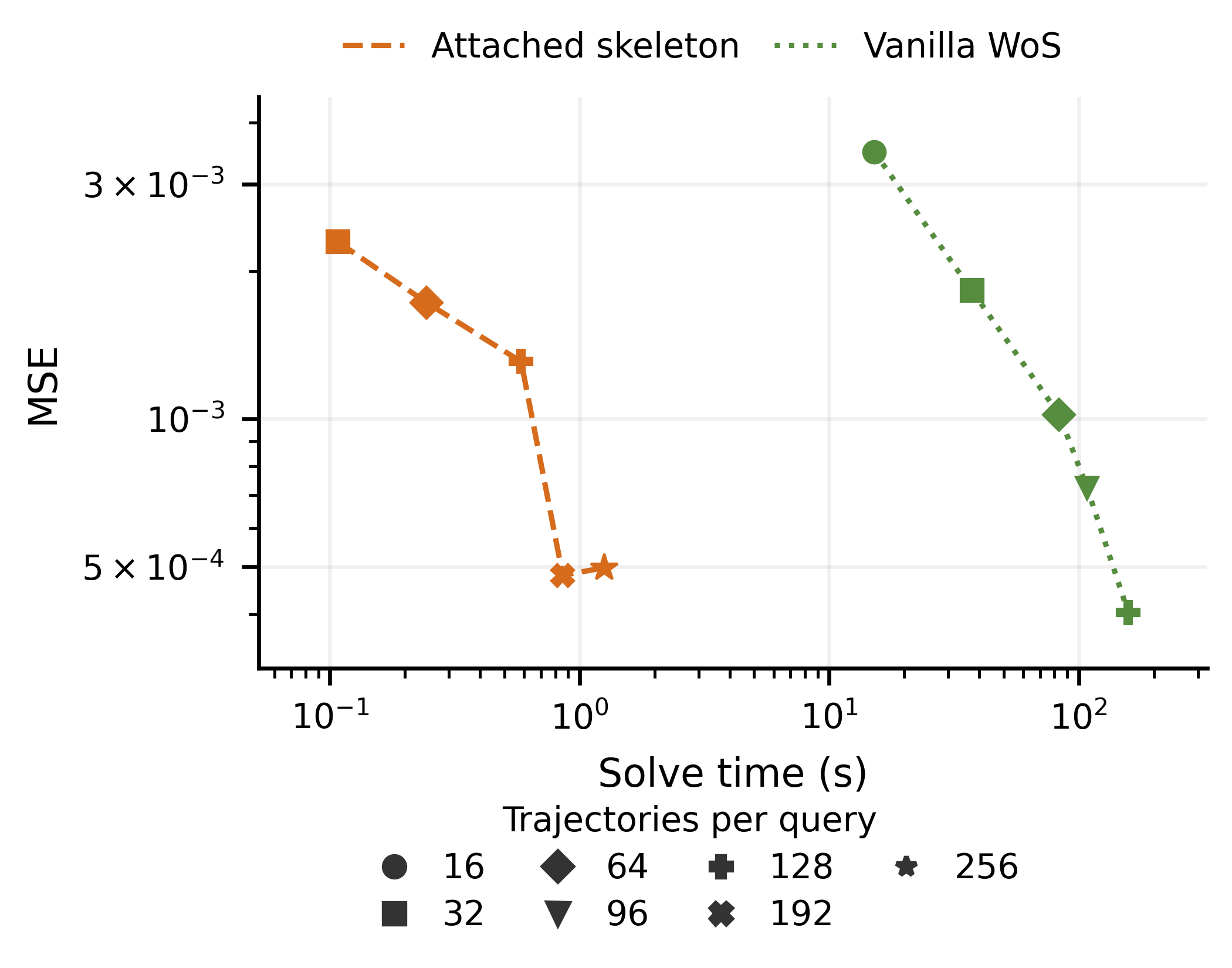}
  \caption{Accuracy and solve time with source terms, averaged over
  random seeds. Marker shape gives trajectories per query, applied to both
  skeleton construction stages. Both axes are logarithmic; construction and
  solver preparation costs are excluded.}
  \Description{Accuracy--solve-time curves for the attached skeleton and
  Vanilla WoS with source terms. The skeleton solves substantially faster;
  its error decreases across most budgets and levels off at the largest
  tested budget.}
  \label{fig:source-accuracy}
\end{figure}

Figure~\ref{fig:source-accuracy} shows that the solve-time advantage
extends to source terms. With 192 trajectories per query, the attached
skeleton achieves 34\% lower MSE than 96-trajectory Vanilla WoS while
solving about $126\times$ faster. Storing volume contributions alongside
boundary transport avoids repeating geometry traversal and source sampling
during evaluation. Increasing the skeleton budget from 192 to 256 does
not further reduce mean error, whereas Vanilla WoS continues to improve.
Thus, additional construction trajectories offer diminishing accuracy
gains in this range despite the low evaluation cost.

\subsection{Variance-Guided Query Refinement}
\label{sec:query-refinement-experiment}

Query refinement can be useful in a domain with localized regions of
high variance, for example, surrounding cups with rapidly varying boundary data
on a few of their interiors, as shown in Figure~\ref{fig:scene-refinement}. Holding the base skeleton fixed, we compare our method with adding the continuation
directly at the queries with the same budgets. Table~\ref{tab:query-refinement} reports the
average error before and after these updates.

\begin{figure}[!htbp]
  \centering
  \includegraphics[width=\linewidth]{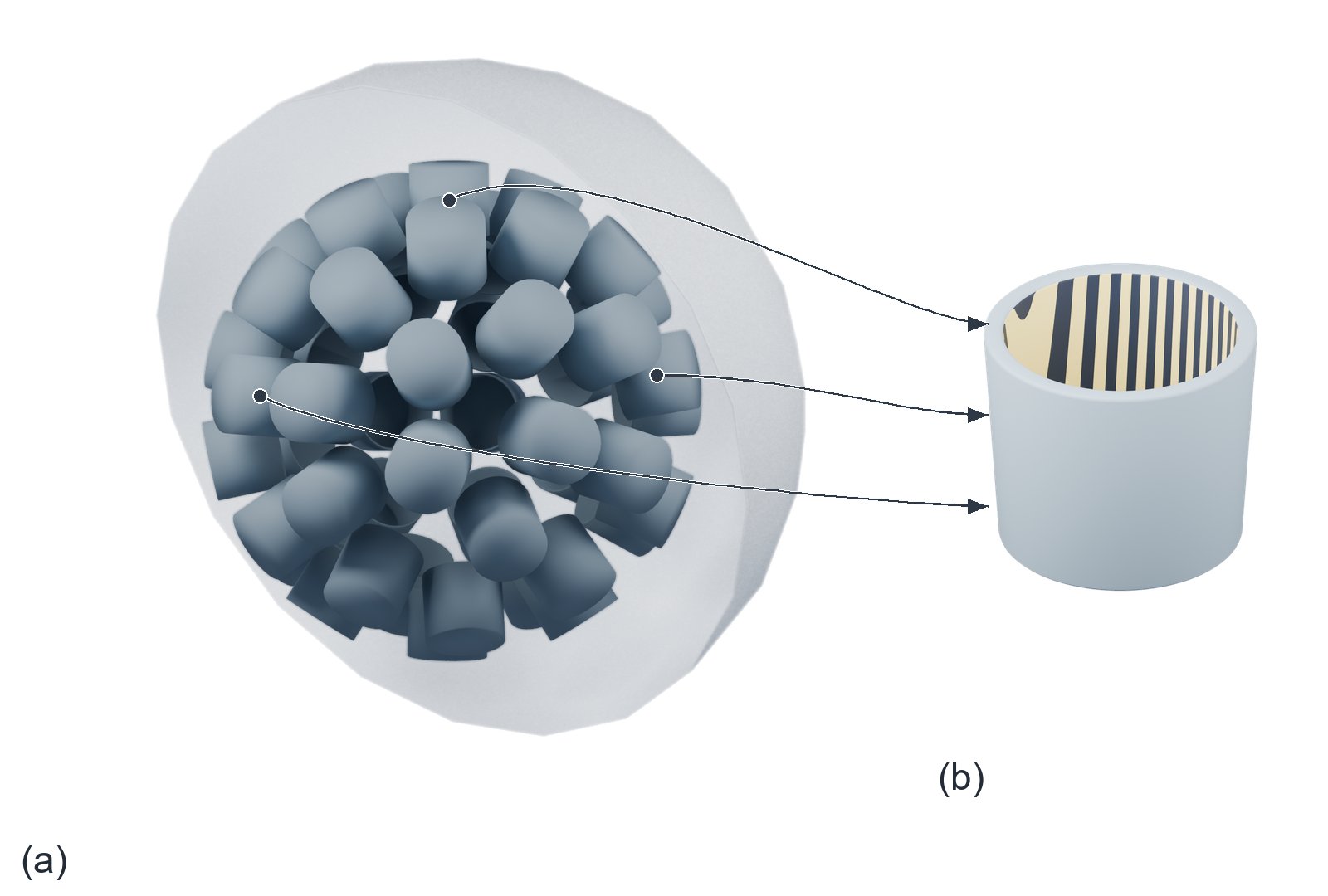}
  \caption{Localized-variance test domain: (a) the complete cup assembly
  inside a enclosure; (b) a cup with high-variance interior. }
  \Description{Naturally lit cups form a spherical assembly within a half-open
  enclosing shell. Arrows point from several cups to a separate schematic
  cup with alternating stripes on its interior.}
  \label{fig:scene-refinement}
\end{figure}

\begin{table}[!htbp]
  \centering
  \caption{Query-averaged MSE for one fixed skeleton. Both updates add
  64 continuations per query.}
  \label{tab:query-refinement}
  \begin{tabular}{lr}
    \toprule
    Method & MSE ($10^{-4}$) \\
    \midrule
    Before refinement & 2.633 \\
    Additional WoS at queries & 1.098 \\
    Variance-guided refinement & 0.637 \\
    \bottomrule
  \end{tabular}
\end{table}

Refinement reduces mean error by 75.8\% from the initial estimate and by
42.0\% relative to adding WoS samples directly at the queries. Targeting
high-variance internal contributions with strong query influence concentrates
the new samples where they can most improve the estimate. Moreover, improving
a shared continuation can benefit several queries, rather than only the query
that receives additional samples. This makes local refinement useful when
uncertainty is spatially concentrated within an otherwise reusable skeleton.

\subsection{Coverage-Preserving Compression}
\label{sec:compression-access-experiment}

To evaluate whether coverage preservation maintains attachment opportunities for new queries, we compare coverage-preserving compression with path-chain elimination on the same skeleton. The required multiplicity in Section~\ref{sec:spatial-compression} is set to $m=1$.

After compression, both graphs add a same batch of queries and perform path attachment with the same sampling budgets. Then we measure the
fraction of all graph nodes removed, attachment success rate, and
mean WoS steps over all new trajectories. Table~\ref{tab:compression-access}
summarizes the trade-off between node reduction and subsequent attachment.

\begin{table}[!htbp]
  \centering
  \caption{Compression and subsequent reuse. Means over matched
  random seeds; step counts are averaged over all new trajectories.}
  \label{tab:compression-access}
  \begin{tabular}{lrr}
    \toprule
    Metric & \shortstack{Path-chain\\elimination} & \shortstack{Coverage-preserving\\compression} \\
    \midrule
    Nodes removed (\%) & 95.53 & 94.77 \\
    Success (\%) & 54.3 & 63.9 \\
    Mean WoS steps & 15.61 & 12.23 \\
    \bottomrule
  \end{tabular}
\end{table}

Preserving access costs only 0.76\% of the original node count, yet reduces
subsequent exploration by 21.7\%. Aggressively contracting every eligible
chain removes few additional nodes at the expense of useful spatial access points.
Spatial anchors therefore preserve much of the compression benefit while
making the retained transport easier for new queries to reach.

\subsection{Spatial Coverage from Path Refinement}
\label{sec:path-refinement-coverage-experiment}

To evaluate how path refinement in Section~\ref{sec:path-refinement} expands the region satisfying the attachment bound, we apply it to a single sampling ball and vary the number of refinement paths. Coverage is measured as the fraction of the ball’s volume satisfying the bound.

Within a retained sampling ball, path refinement can substantially expands the
attachable region, as shown in
Table~\ref{tab:path-refinement-coverage}. 

With eight refinement paths, coverage increases from 6.56\% to 34.69\%,
reaching approximately 5.3 times the original coverage. With 32 paths,
coverage reaches 48.15\%, approximately 7.3 times the original coverage.

\begin{table}[!htbp]
  \centering
  \caption{Attachable volume within a sampling ball, averaged over random
  seeds and measured with uniform interior probes. Each refinement path
  has at most 32 steps, with zero paths representing the original coverage.}
  \label{tab:path-refinement-coverage}
  \footnotesize
  \setlength{\tabcolsep}{2pt}
  \renewcommand{\arraystretch}{0.95}
  \resizebox{\columnwidth}{!}{%
    \begin{tabular}{@{}l*{12}{r}@{}}
      \toprule
      Refinement paths & 0 & 1 & 2 & 4 & 8 & 16 & 32 & 64 & 128 & 256 & 512 & 1024 \\
      \midrule
      Coverage (\%) & 6.56 & 13.67 & 20.30 & 27.36 & 34.69 & 40.72 & 48.15 & 53.68 & 59.52 & 65.14 & 70.22 & 75.18 \\
      \bottomrule
    \end{tabular}%
  }
\end{table}

\begin{figure}[!htbp]
  \centering
  \includegraphics[width=0.78\linewidth]{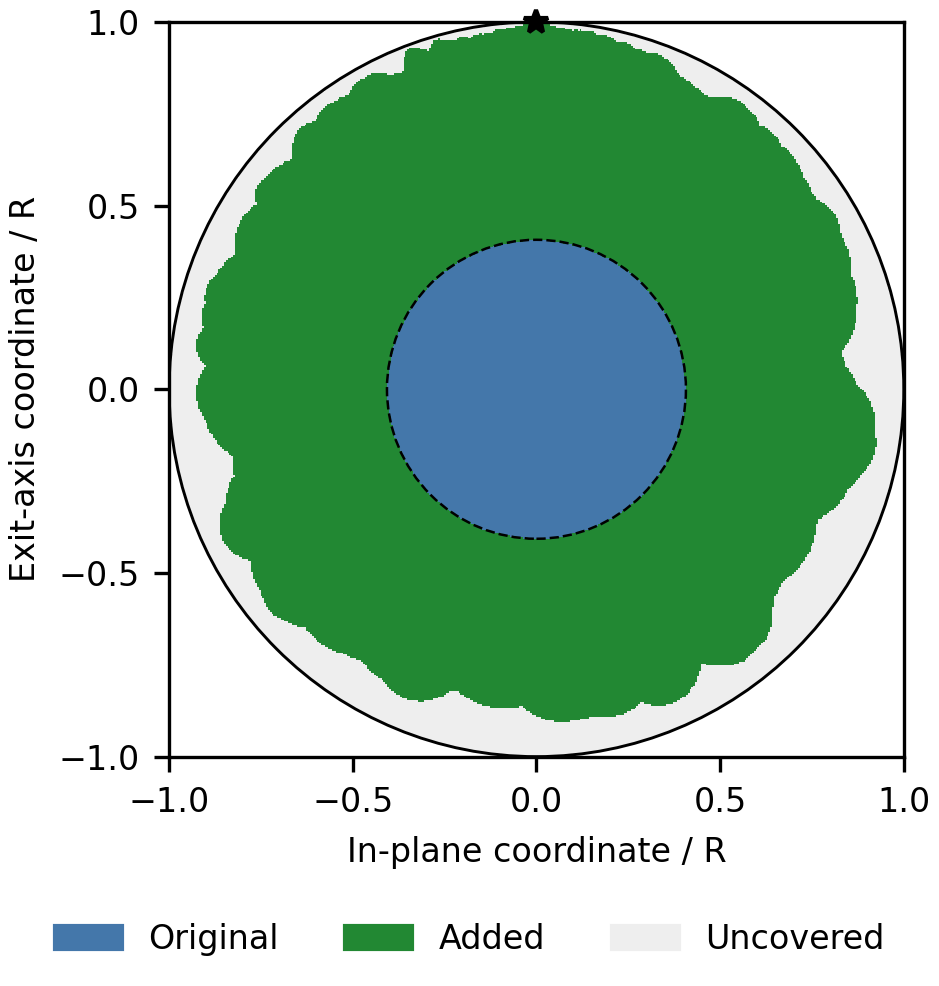}
  \caption{A representative cross-section through the ball center and
  retained exit (star), after 1024 refinement paths. Blue denotes original
  coverage, green denotes added coverage, and gray denotes uncovered space.}
  \Description{A planar section of a ball shows the original central
  attachment region surrounded by a much larger region made attachable by
  path refinement.}
  \label{fig:path-refinement-coverage-slice}
\end{figure}

\subsection{Dense Queries through a Hub}
\label{sec:hub-memory-experiment}

To assess the memory savings of hub evaluation for dense queries, we compare it with ordinary path attachment on the same base skeleton. Queries are added within a hub inside an irregular domain (Figure~\ref{fig:scene-hub}). Figure~\ref{fig:hub-memory} reports the additional persistent memory as the query count increases.

\begin{figure}[!htbp]
  \centering
  \includegraphics[width=0.9\linewidth]{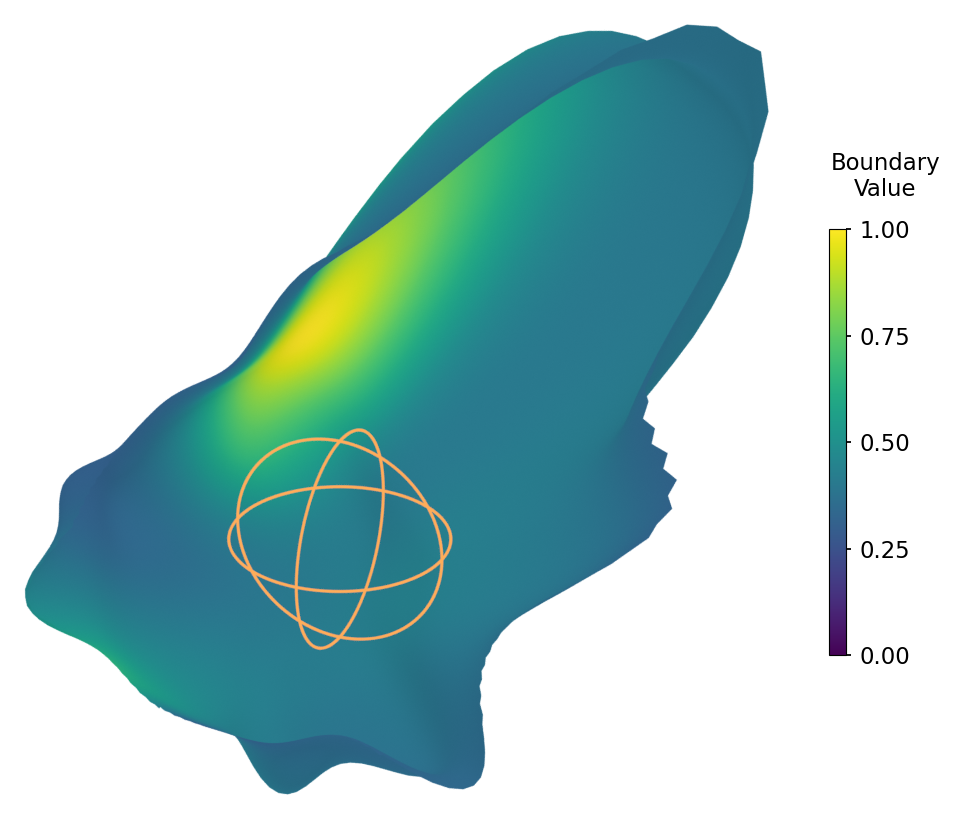}
  \caption{Hub test domain shown in cutaway. Orange circles delineate the
  hub sphere.}
  \Description{A cutaway irregular closed mesh surrounds an orange wire
  sphere marking the Hub. Individual query points are not displayed.}
  \label{fig:scene-hub}
\end{figure}

\begin{figure}[!htbp]
  \centering
  \includegraphics[width=\linewidth]{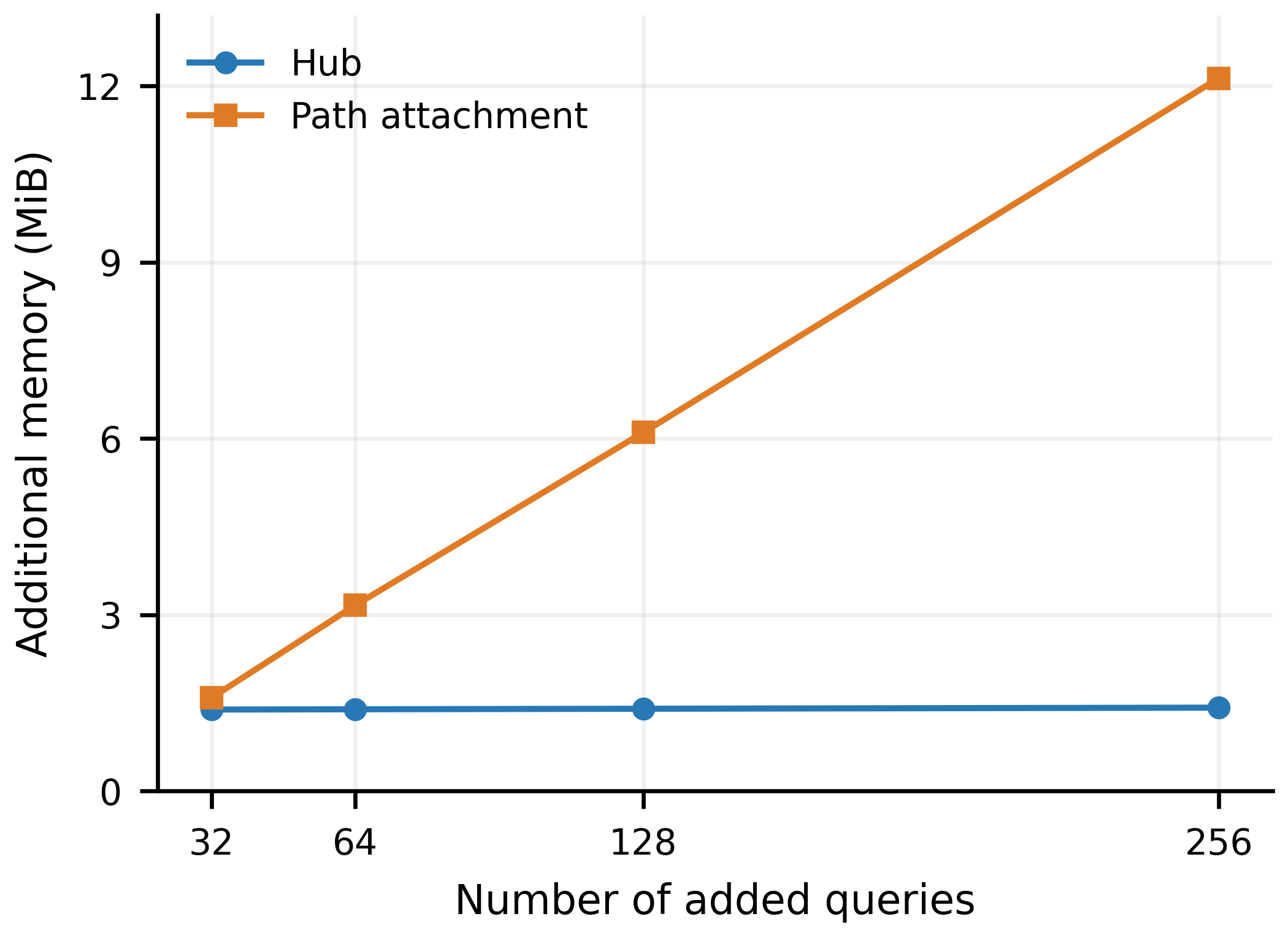}
  \caption{Additional persistent memory for newly added queries, excluding
  the shared base skeleton. Means over random seeds; the hub uses one
  fixed interface pool at all query counts.}
  \Description{Ordinary attachment memory grows approximately linearly with
  query count, while hub memory remains nearly constant.}
  \label{fig:hub-memory}
\end{figure}

The shared pool makes hub memory nearly independent of the number of added
queries over the tested range. Its fixed interface cost is amortized across
more queries, so the storage advantage increases with query density.

Despite substantial memory savings, approaching ordinary attachment's accuracy requires a larger  sample ammount. In the  trials,
an average of 3893 shared interfaces gives a mean MSE of
$1.224\times10^{-3}$, compared with $1.094\times10^{-3}$ for ordinary
attachment using 32 trajectories per query. The hub thus needs more interface samples
to compensate for shared downstream information.

\subsection{Local Repair of Dynamic Transport}
\label{sec:local-repair-experiment}

To test whether local repair preserves important continuations after a geometry change, we construct a dumbbell-shaped domain with two spherical chambers connected by a passage (Figure~\ref{fig:scene-repair}(a)). An initial skeleton is built and refined in the chamber with rapidly varying boundary data, referred to as the remote chamber. Target queries are then added in the other chamber. A small fraction of their walks traverse the passage and attach to the well-explored skeleton in the remote chamber. Finally, the obstacle shown in Figure~\ref{fig:scene-repair}(b) moves to the passage entrance, invalidating sampling balls along these connecting paths.

\begin{figure}[!htbp]
  \centering
  \includegraphics[width=\linewidth]{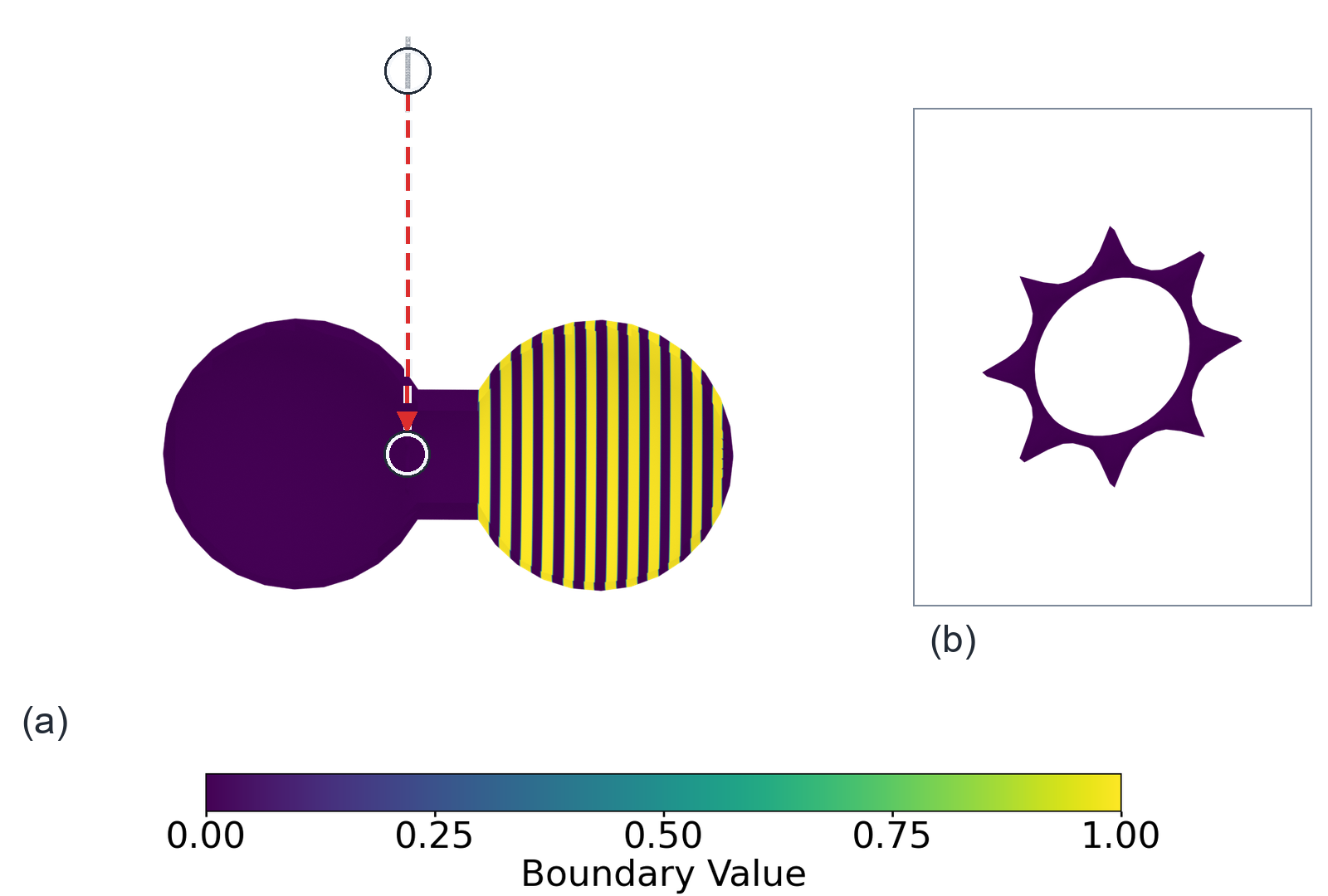}
  \caption{Local-repair setting: (a) two chambers connected by a
  passage, with varying boundary data in the remote chamber; (b) the perforated
  obstacle in detail. The arrow indicates insertion near the passage.
  The enclosure is shown in cutaway.}
  \Description{An arrow moves a thin toothed ring from its initial position
  into a two-chamber domain. A magnified view reveals its central opening.}
  \label{fig:scene-repair}
\end{figure}

We then compare conditional repair with
suffix resampling, which replaces each affected continuation with a fresh WoS
walk from its starting node. Both methods update the solution by reevaluating
the affected ancestors. Among the reachable remote nodes, we measure the fraction that remain reachable afterward. This gives the retention rate of the existing downstream skeleton.
Update time includes repair and ancestor reevaluation,
excluding the shared initial construction. Table~\ref{tab:local-repair}
reports mean results. Target-query MSE isolates the
queries that rely on the remote structure.

\begin{table}[!htbp]
  \centering
  \caption{Retention, update cost, and accuracy after a local geometry change.
  Means over matched random seeds.}
  \label{tab:local-repair}
  \begin{tabular}{lrr}
    \toprule
    Metric & \shortstack{Suffix\\resampling} & \shortstack{Conditional\\repair} \\
    \midrule
    Retention rate (\%) & 0.0 & 67.0 \\
    Update time (s) & 25.16 & 27.19 \\
    Target-query MSE ($10^{-5}$) & 2.009 & 0.934 \\
    \bottomrule
  \end{tabular}
\end{table}

Conditional repair retains access to much of the existing remote transport,
rather than replacing it when an earlier sampling ball becomes invalid.
This reduces target-query MSE by 53.5\% for an 8.0\% increase in update time.

\subsection{Bidirectional Support for Local Repair}
\label{sec:bidirectional-repair-experiment}

We next consider a small obstacle with rapidly varying boundary values and a large contribution relative to its low-valued surroundings, as shown in Figure~\ref{fig:scene-biwos}. Its spiked geometry is chosen to invalidate many sampling balls while remaining difficult for walks to hit. Conditional repair is augmented with Bi-WoS to better resolve the obstacle’s influence while retaining access to old continuations. Table~\ref{tab:bidirectional-repair} compares this extension with conditional repair alone and suffix resampling, using the same initial skeleton and query positions in matched trials.

\begin{figure}[!htbp]
  \centering
  \includegraphics[width=\linewidth]{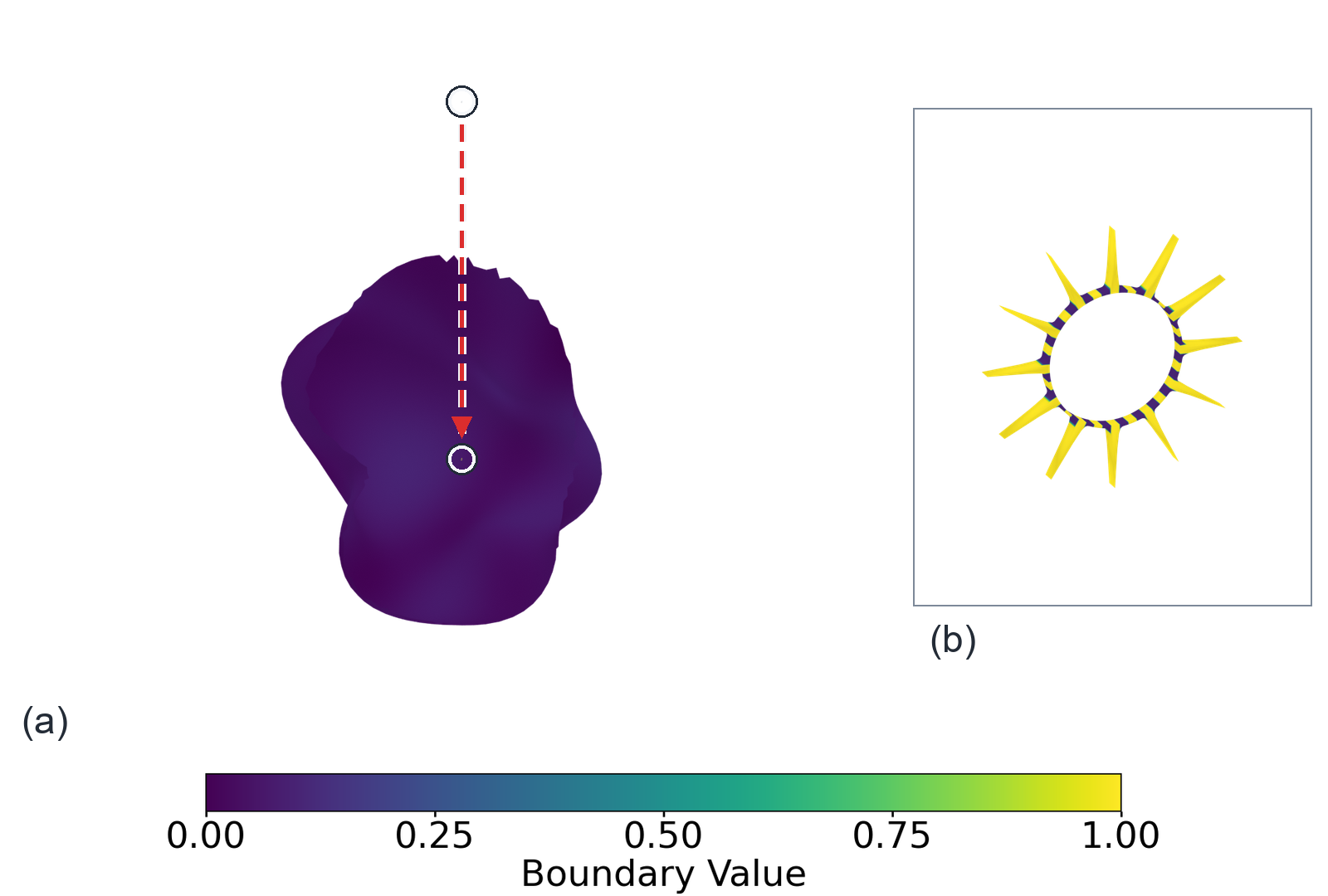}
  \caption{Small-object setting: (a) insertion into the cutaway domain;
  (b) a magnified view of the thin, spiked obstacle and its boundary data.
  The obstacle retains its actual size in the overview.}
  \Description{An insertion arrow locates a very small obstacle in a low-valued
  irregular enclosure; the detail shows its thin ring and sharp spikes.}
  \label{fig:scene-biwos}
\end{figure}

\begin{table}[!htbp]
  \centering
  \caption{All-query accuracy for a difficult-to-hit, high-contribution
  inserted object. Means over matched random seeds.}
  \label{tab:bidirectional-repair}
  \begin{tabular}{lr}
    \toprule
    Method & MSE ($10^{-5}$) \\
    \midrule
    Suffix resampling & 6.644 \\
    Conditional repair & 8.778 \\
    Conditional repair + Bi-WoS & 3.694 \\
    \bottomrule
  \end{tabular}
\end{table}

Adding Bi-WoS reduces mean MSE by 57.9\% relative to conditional repair alone
and by 44.4\% relative to suffix resampling. The conditional-only result shows
that preserving old transport does not by itself resolve a newly introduced
contribution that is rarely encountered by forward paths. Reverse sampling
complements this structural reuse: the repair can retain useful old
continuations while better resolving the influence of a small, important
boundary component. 

Together, the two experiments demonstrate local editing
that preserves existing transport, with bidirectional support improving its
accuracy in a difficult sampling regime.

\subsection{Regional Transport Editing}
\label{sec:regional-edit-experiment}

We demonstrate selective transport editing by amplifying the influence of
the volcano island's boundary on queries around the village island.
Figure~\ref{fig:regional-edit} compares the original solution with three settings of the interaction scaling factor$\gamma$. All four panels use the same
skeleton, query positions, and regional low-rank patch.

Increasing $\gamma$ strengthens the volcano's contribution at the village
queries, while query values in the other regions remain exactly unchanged.
Once the regional patch has been extracted, changing $\gamma$ requires neither
new trajectories nor reconstruction of the skeleton. The retained graph thus
serves as an editable transport operator, allowing independent control of a
selected regional influence while preserving the remaining interactions.

\begin{figure*}[!t]
  \centering
  \captionsetup[subfigure]{font=small,skip=2pt}
  \begin{subfigure}[t]{0.49\textwidth}
    \centering
    \includegraphics[width=\linewidth,trim=160bp 205bp 155bp 155bp,clip]{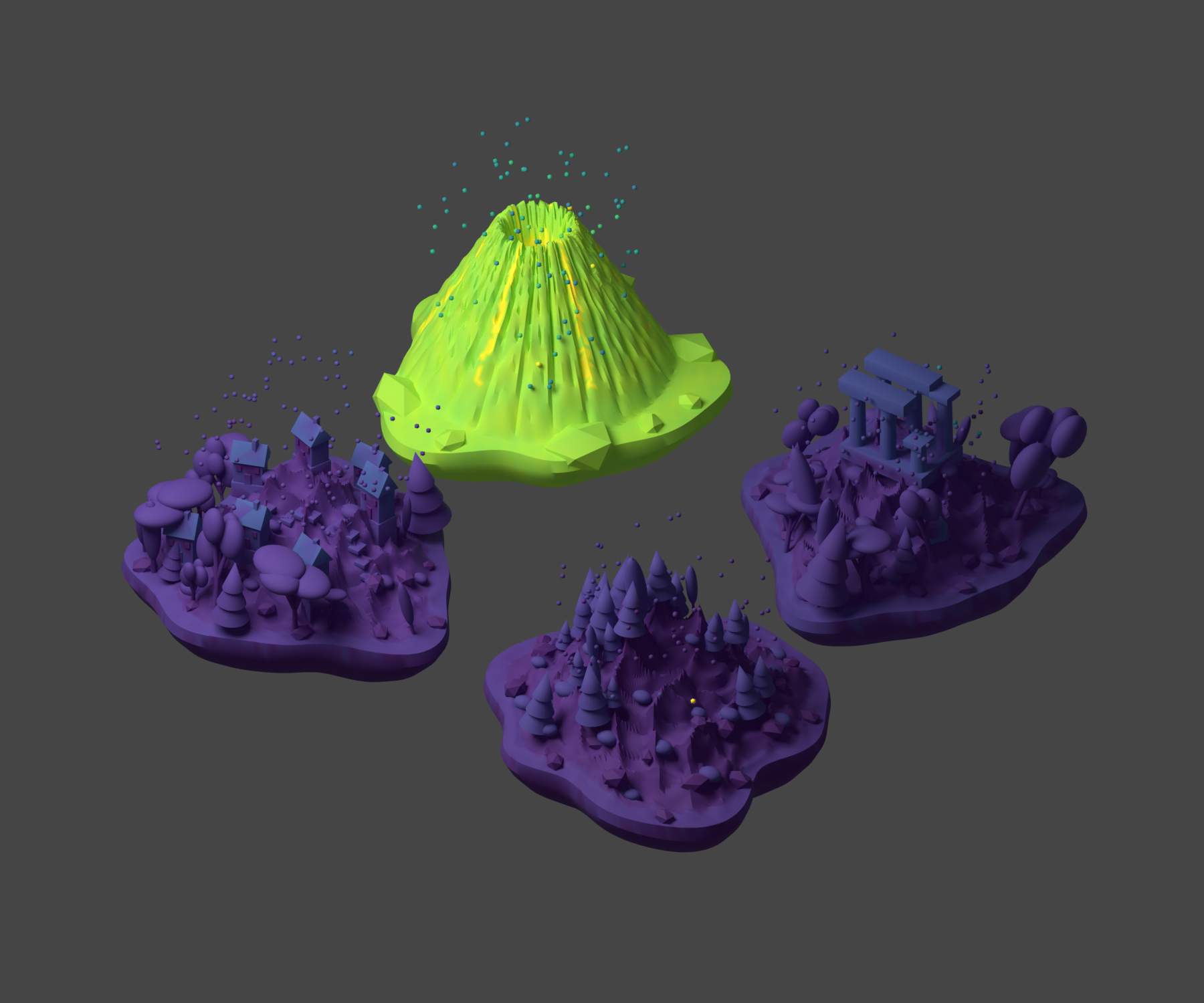}
    \caption{Original ($\gamma=1$)}
    \label{fig:regional-edit-original}
  \end{subfigure}\hfill
  \begin{subfigure}[t]{0.49\textwidth}
    \centering
    \includegraphics[width=\linewidth,trim=160bp 205bp 155bp 155bp,clip]{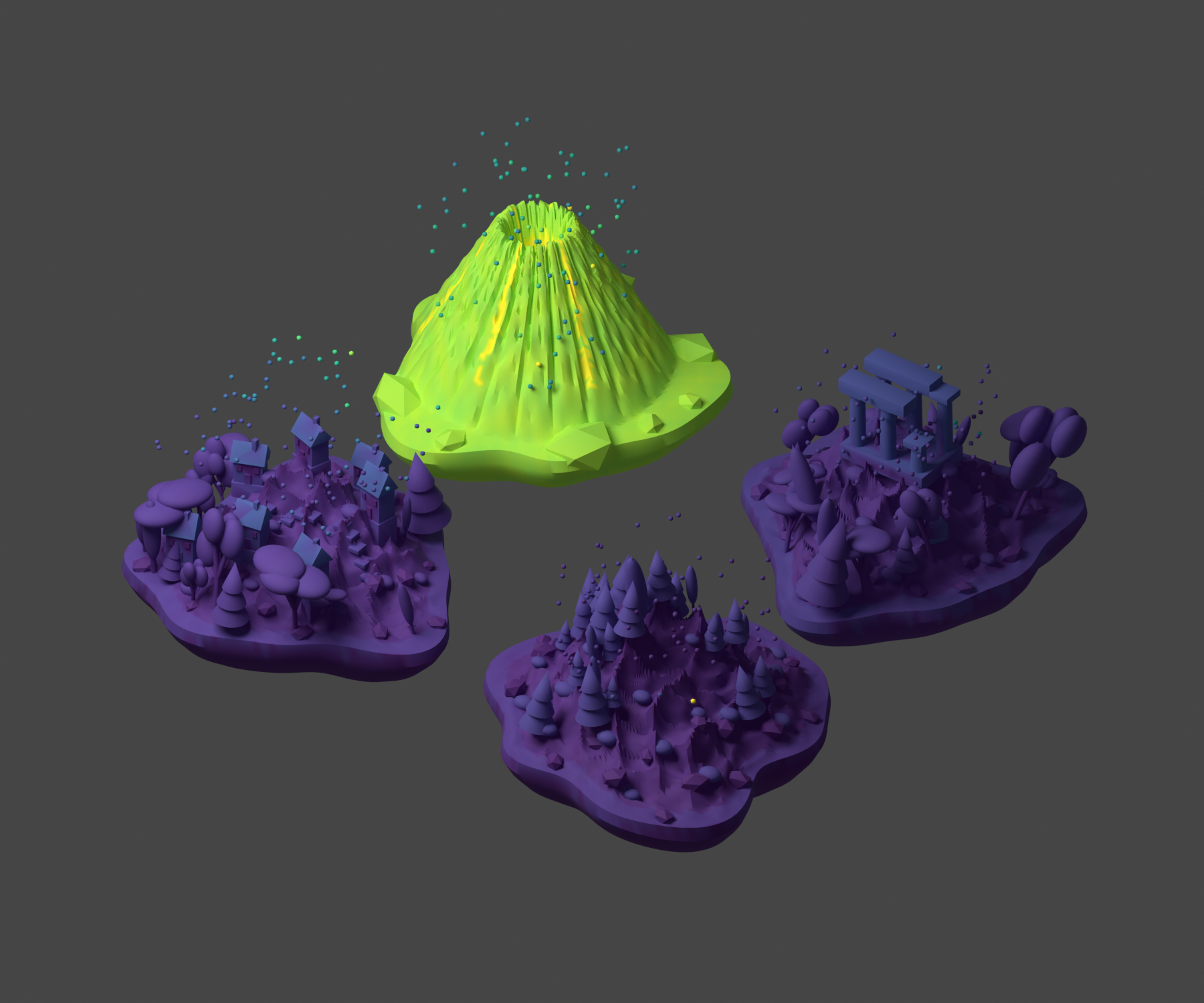}
    \caption{$\gamma=5$}
    \label{fig:regional-edit-gain5}
  \end{subfigure}
  \par\smallskip
  \begin{subfigure}[t]{0.49\textwidth}
    \centering
    \includegraphics[width=\linewidth,trim=160bp 205bp 155bp 155bp,clip]{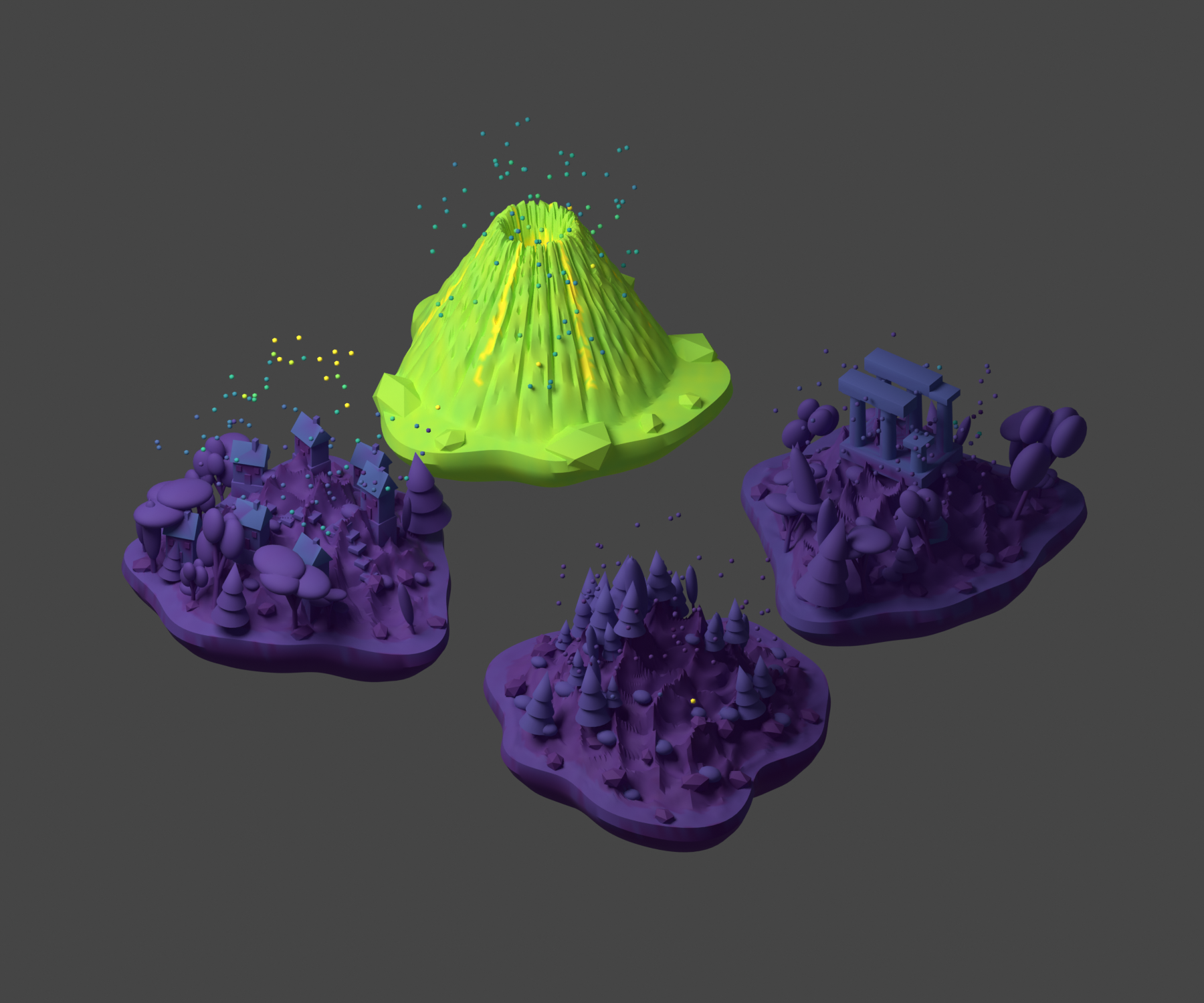}
    \caption{$\gamma=10$}
    \label{fig:regional-edit-gain10}
  \end{subfigure}\hfill
  \begin{subfigure}[t]{0.49\textwidth}
    \centering
    \includegraphics[width=\linewidth,trim=160bp 205bp 155bp 155bp,clip]{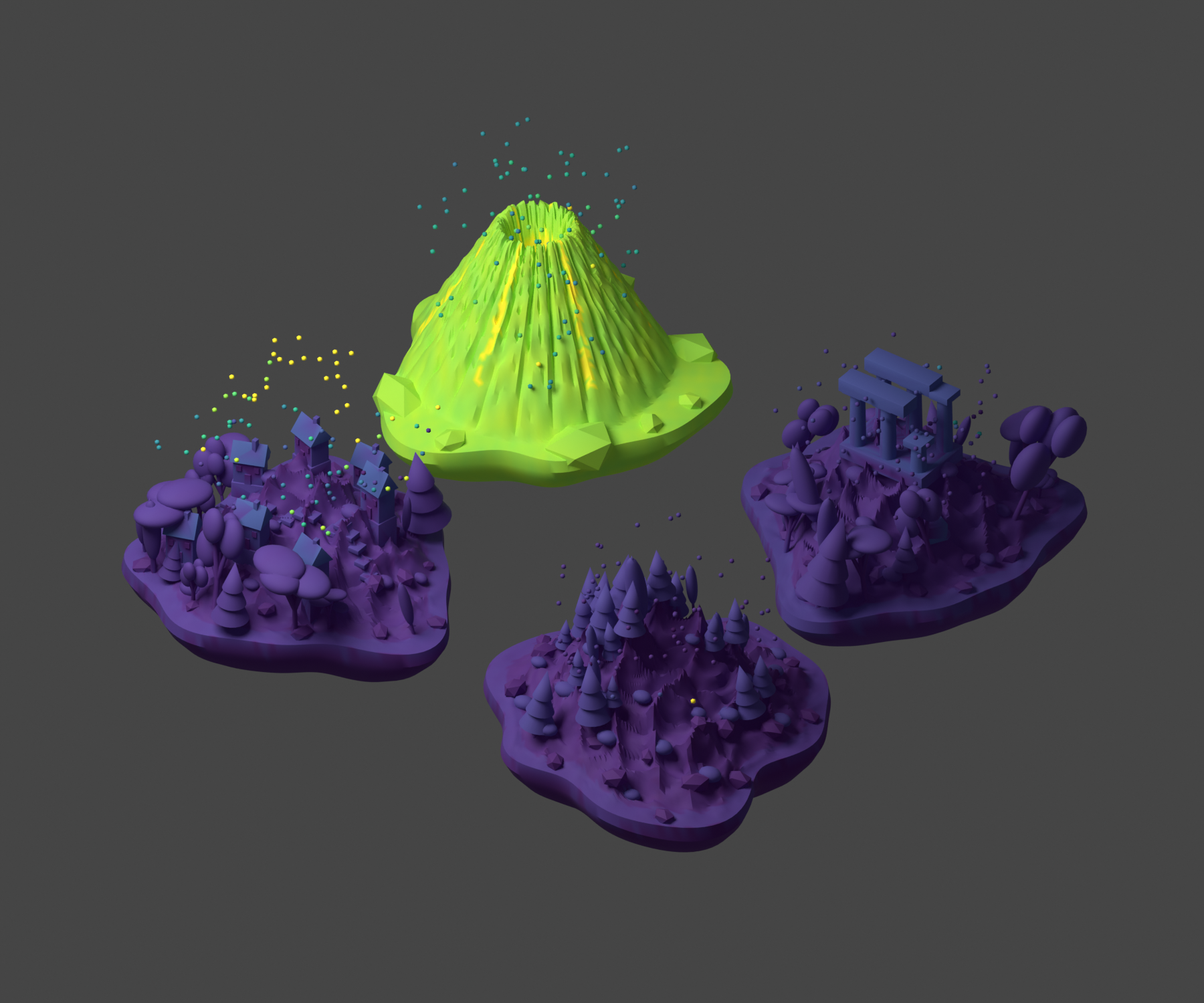}
    \caption{$\gamma=20$}
    \label{fig:regional-edit-gain20}
  \end{subfigure}
  \par\smallskip
  \includegraphics[width=0.65\textwidth]{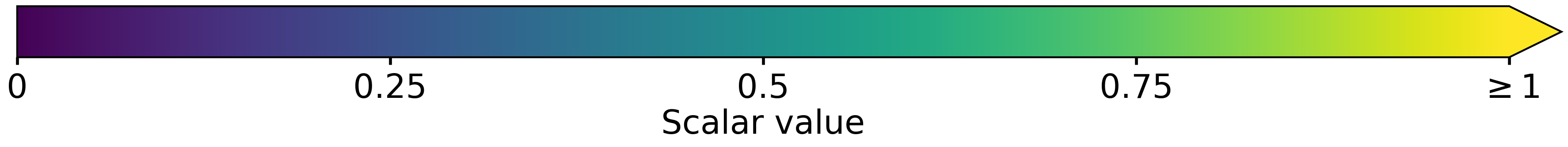}
  \caption{Selective amplification of volcano-to-village transport. Small
  spheres show query values; island surfaces show the fixed boundary values.
  Only the influence of the volcano boundary on village queries is scaled.
  All panels share the same view and viridis color scale; values above 1 use
  the upper endpoint color.}
  \Description{Four views of the same four-island scene on a dark gray background.
  As gamma increases from 1 to 5, 10, and 20, query colors around the village
  shift toward green and yellow while queries around the other islands retain
  their original values.}
  \label{fig:regional-edit}
\end{figure*}


\bibliographystyle{ACM-Reference-Format}
\bibliography{sample-base}


\appendix

\section{Analysis of Guided Neighborhood Attachment}
\label{app:guided-attachment}

We analyze a repeated-attempt model of guided neighborhood attachment to
separate reaching a retained node from controlling the accumulated
importance weights.

\subsection{Repeated-Attempt Model}

Consider successive attachment attempts with uniform target density
$P(\omega)=\frac{1}{4\pi}$ and proposal
\[
Q_n(\omega)
=
\frac{\alpha}{4\pi}
+(1-\alpha)
\frac{\kappa e^{\kappa\mu_n^\top\omega}}
     {4\pi\sinh\kappa},
\qquad 0<\alpha<1,\quad\kappa>0.
\]
Here $\omega\in\mathbb S^2$ is a unit direction, densities are with
respect to surface area on the unit sphere, $\alpha$ is the uniform
mixture fraction, and $\kappa$ is the von Mises--Fisher concentration.
The unit guiding axis $\mu_n$ and success cap
$A_n=\{\omega:\mu_n^\top\omega\ge1-2h_n\}$
are determined by the preceding history.
Here $h_n=P(A_n)$; for a target on a valid sphere
of radius $R_n$, an $\varepsilon$-neighborhood
gives $h_n=\frac{\varepsilon^2}{4R_n^2}$.
Let $T$ denote the first successful attempt and $\mathcal F_{n-1}$ the
history before attempt $n$.

\subsection{Attachment-Time Tail}
\label{app:attachment-tail-proof}

\begin{proposition}[Attachment-time tail]
\label{prop:attachment-tail}
If $h_n\ge h_->0$ before attachment, then
$\mathbb P_Q(T>N)\le(1-s_-)^N\le e^{-Ns_-}$ and $\mathbb E_Q[T]\le s_-^{-1},$ where $s_-=\alpha h_-+ (1-\alpha)\frac{1-e^{-2\kappa h_-}}{1-e^{-2\kappa}}.$
\end{proposition}

\begin{proof}
Conditioned on the history before attempt $n$, the guiding axis $\mu_n$
and cap
\[
A_n=\{\omega:\mu_n^\top\omega\ge1-2h_n\}
\]
are fixed. Since $P(A_n)=h_n$, the uniform component of $Q_n$ contributes
$\alpha h_n$ to the success probability.

For the von Mises--Fisher component, write
$t=\mu_n^\top\omega$. Under the rotationally symmetric surface measure on
$\mathbb S^2$, integration over the cap gives
\begin{align}
\int_{A_n}
\frac{\kappa e^{\kappa\mu_n^\top\omega}}
     {4\pi\sinh\kappa}\,d\omega
&=
\frac{\kappa}{2\sinh\kappa}
\int_{1-2h_n}^{1}e^{\kappa t}\,dt \\
&=
\frac{e^\kappa-e^{\kappa(1-2h_n)}}
     {e^\kappa-e^{-\kappa}} \\
&=
\frac{1-e^{-2\kappa h_n}}
     {1-e^{-2\kappa}} .
\end{align}
Therefore the conditional success probability is
\[
s(h_n)
=
\alpha h_n
+
(1-\alpha)
\frac{1-e^{-2\kappa h_n}}
     {1-e^{-2\kappa}} .
\]
The function $s(h)$ is increasing in $h$, so
$h_n\ge h_-$ implies
\[
\mathbb Q(T=n\mid T\ge n,\mathcal F_{n-1})
=
s(h_n)
\ge
s_-,
\]
where
\[
s_-
=
\alpha h_-
+
(1-\alpha)
\frac{1-e^{-2\kappa h_-}}
     {1-e^{-2\kappa}} .
\]
Hence
\[
\mathbb Q(T>n\mid T>n-1,\mathcal F_{n-1})
\le
1-s_- .
\]
Iterating conditional expectations gives
\[
\mathbb Q(T>N)
\le
(1-s_-)^N
\le
e^{-Ns_-}.
\]

Finally, for a positive integer-valued stopping time,
\[
\mathbb E_Q[T]
=
\sum_{N=0}^{\infty}\mathbb Q(T>N)
\le
\sum_{N=0}^{\infty}(1-s_-)^N
=
\frac1{s_-}.
\]
\end{proof}

\subsection{Stopped-Weight Divergence}
\label{app:weight-divergence-proof}

An exponential attachment-time tail does not guarantee a finite second
moment of the accumulated likelihood ratio.

\begin{proposition}[Stopped-weight divergence]
\label{prop:weight-divergence}
Define $\ell_n=\frac{P(\omega_n)}{Q_n(\omega_n)}$ and
$W_n=\prod_{j=1}^{n}\ell_j$. If
$0<h_-\le h_n\le h_+<1-\alpha$ before attachment,
then, for sufficiently large fixed $\kappa$, $\mathbb E_Q[W_T]=1$ but $\mathbb E_Q[W_T^2]=+\infty$, despite the exponential tail in
Proposition~\ref{prop:attachment-tail}.
\end{proposition}

\begin{proof}
Condition on the history $\mathcal F_{n-1}$ before attempt $n$, so that
$\mu_n$, $h_n$, and $A_n$ are fixed. Write
$t=\mu_n^\top\omega$. Relative to the uniform target density $P$, the
proposal satisfies
\[
\frac{Q_n(\omega)}{P(\omega)}
=
D_\kappa(t)
:=
\alpha
+
(1-\alpha)
\frac{\kappa e^{\kappa t}}{\sinh\kappa},
\]
and therefore
\[
\ell_n(\omega)
=
\frac{P(\omega)}{Q_n(\omega)}
=
D_\kappa(t)^{-1}.
\]

Since $t$ is uniform on $[-1,1]$ with density $\frac{1}{2}$ under $P$, for any
measurable set $C$ depending only on $t$,
\[
\mathbb E_Q[
\ell_n^2\mathbf1_C
\mid\mathcal F_{n-1}]
=
\int_C\frac{P(\omega)^2}{Q_n(\omega)}\,d\omega
=
\frac12\int_C\frac{dt}{D_\kappa(t)}.
\]

Before attachment,
$h_n\le h_+$ implies
$1-2h_n\ge1-2h_+$, so the failure region
$\{t<1-2h_n\}$ contains
$[-1,1-2h_+]$. Hence
\[
\mathbb E_Q[
\ell_n^2\mathbf1_{\{\omega_n\notin A_n\}}
\mid\mathcal F_{n-1}]
\ge
a_\kappa,
\]
where
\[
a_\kappa
=
\frac12
\int_{-1}^{1-2h_+}
\frac{dt}{D_\kappa(t)}.
\]
Similarly, $h_n\ge h_-$ implies
$1-2h_n\le1-2h_-$, so the success region contains
$[1-2h_-,1]$, and therefore
\[
\mathbb E_Q[
\ell_n^2\mathbf1_{\{\omega_n\in A_n\}}
\mid\mathcal F_{n-1}]
\ge
b_\kappa,
\]
with
\[
b_\kappa
=
\frac12
\int_{1-2h_-}^{1}
\frac{dt}{D_\kappa(t)}
>0.
\]

For every fixed $t<1$,
\[
\frac{\kappa e^{\kappa t}}{\sinh\kappa}
\longrightarrow 0
\qquad
(\kappa\to\infty).
\]
Moreover,
$D_\kappa(t)\ge\alpha$, hence
$D_\kappa(t)^{-1}\le\alpha^{-1}$.
Dominated convergence therefore gives
\[
a_\kappa
\longrightarrow
\frac12
\int_{-1}^{1-2h_+}
\frac{dt}{\alpha}
=
\frac{1-h_+}{\alpha}.
\]
Since $h_+<1-\alpha$,
\[
\frac{1-h_+}{\alpha}>1.
\]
Thus, for all sufficiently large fixed $\kappa$,
$a_\kappa>1$.

On the event $\{T=n\}$, the first $n-1$ attempts fail and the $n$th
succeeds. Repeated conditioning with the bounds above yields
\[
\mathbb E_Q[
W_n^2\mathbf1_{\{T=n\}}
]
\ge
a_\kappa^{\,n-1}b_\kappa.
\]
Consequently,
\begin{align}
\mathbb E_Q[W_T^2]
&=
\sum_{n=1}^{\infty}
\mathbb E_Q[
W_n^2\mathbf1_{\{T=n\}}
] \\
&\ge
b_\kappa
\sum_{n=1}^{\infty}
a_\kappa^{\,n-1}
=
+\infty .
\end{align}

It remains to establish the first moment.
Under the target law $P$, the conditional probability of success at every
pre-attachment step equals $h_n\ge h_-$.
Hence
\[
\mathbb P_P(T>N)
\le
(1-h_-)^N,
\]
so $T<\infty$ almost surely under $P$.

The likelihood ratio satisfies
\[
\mathbb E_Q[
\ell_n\mid\mathcal F_{n-1}]
=
\int Q_n(\omega)
\frac{P(\omega)}{Q_n(\omega)}
\,d\omega
=
1,
\]
and, similarly,
\begin{align}
\mathbb E_Q[
\ell_n
\mathbf1_{\{\omega_n\notin A_n\}}
\mid\mathcal F_{n-1}]
&=
P(A_n^c) \\
&=
1-h_n \\
&\le
1-h_-.
\end{align}
Therefore, by repeated conditioning,
\[
\mathbb E_Q[
W_N\mathbf1_{\{T>N\}}
]
\le
(1-h_-)^N.
\]

For every finite $N$, partitioning according to whether the stopping event
has occurred gives
\[
1
=
\sum_{n=1}^{N}
\mathbb E_Q[
W_n\mathbf1_{\{T=n\}}
]
+
\mathbb E_Q[
W_N\mathbf1_{\{T>N\}}
].
\]
Letting $N\to\infty$, the final term vanishes, and monotone convergence of
the nonnegative partial sums gives
\[
\mathbb E_Q[W_T]
=
\sum_{n=1}^{\infty}
\mathbb E_Q[
W_n\mathbf1_{\{T=n\}}
]
=
1.
\]
Thus
\[
\mathbb E_Q[W_T]=1,
\qquad
\mathbb E_Q[W_T^2]=+\infty.
\]
\end{proof}

\section{Proofs and Implementation Details for Path Refinement}
\label{app:path-refinement}

\subsection{Correctness of Conditional Refinement}
\label{app:refinement-proof}

Let $B=B(p,R)$ and let $Q=X_\tau$ denote the first exit of Brownian
motion from $B$. For an intermediate stopping point $P_k$ generated before
$\tau$, the strong Markov property gives
\begin{equation}
\mathbb P_p\{P_k\in dx,Q\in dA_q\}
=
\mathbb P_p\{P_k\in dx\}
K_B(x,q)\,dA_q .
\label{eq:app-refinement-joint}
\end{equation}
Consequently,
\begin{equation}
\mathbb P_p\{Q\in dA_q\mid P_k=x\}
=
K_B(x,q)\,dA_q .
\label{eq:app-refinement-conditional-exit}
\end{equation}
If $U_Q$ is conditionally unbiased,
$\mathbb E[U_Q\mid Q=q]=u(q)$, and refinement samples are generated
independently of its downstream randomness given $Q$, then
\begin{equation}
\mathbb E[U_Q\mid P_k=x]
=
\int_{\partial B}
K_B(x,q)u(q)\,dA_q
=
u(x).
\label{eq:app-refinement-node-semantics}
\end{equation}
Hence every inserted state has the same pointwise continuation semantics
as an ordinary WoS state.

Conditioning the local transition on the known final exit $Q=q$ gives
\begin{equation}
\begin{aligned}
\mathbb P
\{X_{\tau_{B_k}}\in dA_s\mid Q=q\}
&=
\frac{
|\partial B_k|^{-1}dA_s\,
K_B(s,q)dA_q
}{
K_B(p_k,q)dA_q
}\\
&=
\frac{1}{|\partial B_k|}
\frac{K_B(s,q)}
     {K_B(p_k,q)}
dA_s ,
\end{aligned}
\end{equation}
which is Eq.~\eqref{eq:refinement-conditional-density}.
Thus the refined process is exactly the corresponding Doob $h$-transform
with $h(x)=K_B(x,q)$.

For comparison, replacing an ordinary retained edge of coefficient $w$
by a chain $p\to p_1\to\cdots\to p_m\to q$ uses
\begin{equation}
\omega_{p\to p_1}=w,\qquad
\omega_{p_j\to p_{j+1}}=1,\qquad
\omega_{p_m\to q}=1.
\label{eq:refinement-uniform-weights}
\end{equation}
Its realized contribution at $p$ remains $wU_Q$, and every intermediate
node evaluates to $U_Q$. The construction in
Section~\ref{sec:path-refinement} instead preserves the original edge
and connects each intermediate node directly to $q$. Its intermediate
values and the original contribution at $p$ are identical to those of
the replacement chain.

\subsection{Importance-Sampled Endpoints}
\label{app:refinement-endpoints}

Suppose now that $Q$ is sampled from area density $\rho(q)$ and define
\[
\ell(q)=\frac{K_B(p,q)}{\rho(q)}.
\]
Let $M\ge\sup_q\ell(q)$ be fixed before observing $Q$, and let
$A=\{V\le\frac{\ell(Q)}{M}\}$ for an independent
$V\sim\mathrm{Uniform}(0,1)$.
For a replacement chain, accept the endpoint once and use
\begin{equation}
\omega_{p\to p_1}=wM,\qquad
\omega_{p_j\to p_{j+1}}=1,\qquad
\omega_{p_m\to q}=1.
\label{eq:refinement-is-weights}
\end{equation}
A rejected attempt contributes zero and creates no intermediate nodes.
Then
\begin{equation}
\begin{aligned}
\mathbb Q\{A,Q\in dA_q\}
&=
\rho(q)\frac{\ell(q)}{M}\,dA_q\\
&=
\frac1M K_B(p,q)\,dA_q .
\end{aligned}
\end{equation}
Integrating over $q$ yields $\mathbb Q(A)=\frac{1}{M}$, and therefore
\begin{equation}
\mathbb Q\{Q\in dA_q\mid A\}
=
K_B(p,q)\,dA_q .
\label{eq:app-endpoint-correction}
\end{equation}

Conditional refinement given an accepted endpoint follows the same
Brownian-bridge law as above. Using
\[
\mathbb P_p\{P_k\in dx\mid Q=q\}
=
\frac{
\mathbb P_p\{P_k\in dx\}K_B(x,q)
}{
K_B(p,q)
},
\]
we obtain
\begin{equation}
\begin{aligned}
&\mathbb Q\{A,P_k\in dx,Q\in dA_q\}\\
&\quad=
\rho(q)dA_q\,
\frac{K_B(p,q)}{M\rho(q)}
\frac{
\mathbb P_p\{P_k\in dx\}K_B(x,q)
}{
K_B(p,q)
}\\
&\quad=
\frac1M
\mathbb P_p\{P_k\in dx\}
K_B(x,q)dA_q .
\end{aligned}
\end{equation}
Hence
\[
\mathbb Q\{Q\in dA_q\mid A,P_k=x\}
=
K_B(x,q)dA_q,
\]
and therefore
\[
\mathbb E_{\mathbb Q}
[U_Q\mid A,P_k=x]
=
u(x).
\]

The compensating first-edge factor preserves the original branch
contribution:
\begin{equation}
\mathbb E\!\left[
wM\mathbf1_AU_Q\mid Q,U_Q
\right]
=
wM\frac{\ell(Q)}M U_Q
=
w\ell(Q)U_Q .
\end{equation}

The endpoint test must be performed only once per attempt.
Repeating the Bernoulli test for the same fixed endpoint until success
would condition on eventual acceptance and hence restore the original
proposal $\rho$ rather than the target Poisson law.

\subsection{Path Splitting}

Let $A_1,\ldots,A_N$ be conditionally independent endpoint tests for the
same retained $(Q,U_Q)$, each with
\[
\mathbb P(A_n\mid Q)=\frac{\ell(Q)}M .
\]
Accepted replacement chains use first-edge weight $\frac{wM}{N}$.
The split estimator is
\begin{equation}
\widehat C_N
=
\frac{wM}{N}
\sum_{n=1}^{N}\mathbf1_{A_n}U_Q .
\label{eq:refinement-splitting}
\end{equation}
Conditioning on $(Q,U_Q)$ gives
\[
\mathbb E[\widehat C_N\mid Q,U_Q]
=
\frac{wM}{N}
N\frac{\ell(Q)}M U_Q
=
w\ell(Q)U_Q .
\]
Since the indicators are conditionally independent Bernoulli variables,
\begin{align}
\operatorname{Var}
(\widehat C_N\mid Q,U_Q)
&=
\frac{w^2M^2}{N^2}
N
\frac{\ell(Q)}M
\left(1-\frac{\ell(Q)}M\right)
U_Q^2\\
&=
\frac{w^2}{N}
\ell(Q)\bigl(M-\ell(Q)\bigr)U_Q^2 .
\end{align}
Thus splitting reduces only the additional endpoint-rejection variance by
$\frac{1}{N}$; the shared exit $Q$ and downstream estimator $U_Q$ are not
independently resampled.

For ordinary WoS exits, $\ell=M=1$, so every branch is accepted and its
first-edge coefficient becomes $\frac{w}{N}$.

For a desired expected count $\lambda$ of accepted branches, choose a
predetermined integer $N\approx\lambda M$.
The number $N$ must be specified independently of the realized endpoint.
Likewise, replacing $\mathbf1_{A_n}$ by its expectation $\frac{\ell(Q)}{M}$
would preserve the scalar contribution but would retain intermediate
states for every proposal endpoint. Those retained states would then be
selected under the proposal law $\rho$ and would not in general satisfy
the pointwise semantics of
Eq.~\eqref{eq:app-refinement-node-semantics}.

\begin{proposition}[Validity of path refinement]
\label{prop:refinement-validity}
Assume that $U_Q$ is conditionally unbiased given $Q$, and that refinement
samples use independent randomness given $Q$. The retained intermediate
states then satisfy $\mathbb E[U_{P_k}\mid P_k=x]=u(x)$. For
importance-sampled exits, the one-shot endpoint correction preserves the
original expected branch contribution, with
\[
\mathbb E[\widehat C_N\mid Q,U_Q]=w\ell(Q)U_Q,
\qquad
\operatorname{Var}(\widehat C_N\mid Q,U_Q)
=\frac{w^2}{N}\ell(Q)\bigl(M-\ell(Q)\bigr)U_Q^2.
\]
\end{proposition}

\subsection{Exact Sampling of the Conditioned Transition}
\label{app:refinement-sampler}

Let $d$ denote the spatial dimension. For one refinement step, write
\[
x=p_k,\qquad
a=\|x-p\|,\qquad
r=R_k,\qquad
\ell=\|q-x\|,
\]
and choose
\[
0<r\le\eta(R-a),
\qquad 0<\eta<1.
\]
Since $\ell\ge R-a$, this guarantees $r<\ell$.

Using the Poisson kernel of the outer ball $B(p,R)$,
Eq.~\eqref{eq:refinement-conditional-density} becomes
\begin{equation}
f_k(s\mid q)
=
\frac{\ell^d}
{|\partial B_k|(R^2-a^2)}
\frac{R^2-\|s-p\|^2}
{\|s-q\|^d}.
\label{eq:app-conditioned-density}
\end{equation}

We use the proposal
\begin{equation}
g_k(s\mid q)
=
\frac{\|s-q\|^{-d}}{Z_k},
\qquad
Z_k
=
\frac{
|\partial B_k|\ell^{2-d}
}{
\ell^2-r^2
}.
\label{eq:app-distance-proposal}
\end{equation}
The normalization follows from
\[
\int_{\partial B_k}
\|s-q\|^{-d}\,dA_s
=
\frac{
|\partial B_k|\ell^{2-d}
}{
\ell^2-r^2
}.
\]

\paragraph{Two dimensions.}
Let $e_1=\frac{q-x}{\ell}$ and let $e_2$ be perpendicular to $e_1$.
For $U\sim\mathrm{Uniform}(0,1)$, set
\begin{equation}
\theta
=
2\arctan\!\left[
\frac{\ell-r}{\ell+r}
\tan\!\left(\pi(U-\tfrac12)\right)
\right],
\end{equation}
and
\[
s
=
x+r(\cos\theta\,e_1+\sin\theta\,e_2).
\]
This samples the boundary density proportional to
$\|s-q\|^{-2}$.

\paragraph{Three dimensions.}
Draw independent $U_1,U_2\sim\mathrm{Uniform}(0,1)$ and define
\[
t
=
\frac{\ell^2-r^2}
{\ell+r-2rU_1},
\qquad
\cos\theta
=
\frac{\ell^2+r^2-t^2}
{2\ell r},
\qquad
\phi=2\pi U_2 .
\]
Using an orthonormal frame with
$e_3=\frac{q-x}{\ell}$,
\[
s
=
x+r\left(
\sin\theta\cos\phi\,e_1
+
\sin\theta\sin\phi\,e_2
+
\cos\theta\,e_3
\right).
\]
This gives the required density proportional to $\|s-q\|^{-3}$.

\paragraph{Rejection correction.}
Let
\[
A_{\mathrm{geo}}(s)
=
R^2-\|s-p\|^2,
\qquad
A_{\max}
=
R^2-(a-r)^2 .
\]
From Eqs.~\eqref{eq:app-conditioned-density} and
\eqref{eq:app-distance-proposal},
\[
\frac{f_k(s\mid q)}{g_k(s\mid q)}
=
\frac{
\ell^2 A_{\mathrm{geo}}(s)
}{
(R^2-a^2)(\ell^2-r^2)
}.
\]
Hence a proposal $s\sim g_k$ is accepted with probability
\[
\frac{A_{\mathrm{geo}}(s)}{A_{\max}}.
\]
The accepted density is exactly $f_k$, so no likelihood correction is
attached to an accepted local transition.

Its average acceptance probability is
\begin{equation}
P_{\mathrm{acc}}
=
\frac{
(R^2-a^2)(1-r^2/\ell^2)
}{
R^2-(a-r)^2
}.
\end{equation}
Let $\delta=R-a$. Since $\ell\ge\delta$ and
$r\le\eta\delta$,
\begin{align}
P_{\mathrm{acc}}
&\ge
\frac{
(R^2-a^2)(1-r^2/\delta^2)
}{
R^2-(a-r)^2
}\\
&=
\frac{
(R+a)(1-r/\delta)
}{
R+a-r
}
\ge
1-\frac r\delta
\ge
1-\eta .
\end{align}
Thus $\eta=\frac{1}{2}$ gives at most two proposal draws on average.

Local rejection here is fundamentally different from endpoint rejection:
a rejected local candidate is redrawn at the same state until accepted
because this procedure samples the desired conditional transition itself.
The endpoint correction, by contrast, must be executed only once for a
fixed retained exit.

\subsection{Comparison with Independent Proxy Sampling}
\label{app:proxy-comparison}

Suppose the retained outer-sphere exit is $Q\sim r$.
Introducing an algebraic proxy $z\in B$ while preserving the same exit
sample gives
\[
\frac{K_B(y,Q)}{K_B(z,Q)}
\frac{K_B(z,Q)}{r(Q)}U_Q
=
\frac{K_B(y,Q)}{r(Q)}U_Q
=: \widehat u_B(y),
\]
so such a proxy is exactly equivalent to direct reuse.

Now let $\overline C\subset B$, let $y\in C$, and sample
$Z$ uniformly on $\partial C$, independently of $(Q,U_Q)$.
The corresponding estimator is
\[
\widehat u_C(y)
=
|\partial C|K_C(y,Z)
\frac{K_B(Z,Q)}{r(Q)}U_Q .
\]
Conditioning on $(Q,U_Q)$ and using the composition of harmonic measures,
\begin{align}
\mathbb E[\widehat u_C(y)\mid Q,U_Q]
&=
\frac{U_Q}{r(Q)}
\int_{\partial C}
K_C(y,z)K_B(z,Q)\,dA_z\\
&=
\frac{K_B(y,Q)}{r(Q)}U_Q
=
\widehat u_B(y).
\end{align}
Therefore, whenever second moments are finite,
\[
\operatorname{Var}[\widehat u_C(y)]
\ge
\operatorname{Var}
\!\left[
\mathbb E[\widehat u_C(y)\mid Q,U_Q]
\right]
=
\operatorname{Var}[\widehat u_B(y)].
\]
Thus independent proxy sampling provides no generic variance improvement.
Conditional refinement differs because the intermediate state and the
retained exit are sampled jointly, yielding
$Q\mid P_k=x\sim K_B(x,\cdot)$ and therefore legitimate reusable
continuation states.

\subsection{Proof of Hub Reuse}
\label{app:hub-reuse-proof}


\begin{proof}
We first consider a retained interface group $j$.
Let $S_j$ denote its supporting sphere and let the
retained exit $Q_j$ have area density $r_j$ on
$\partial S_j$. By assumption, the referenced
downstream estimator is conditionally unbiased:
\begin{equation}
\mathbb E[U_{Q_j}\mid Q_j=q]
=
u(q).
\label{eq:app-hub-downstream-unbiased}
\end{equation}
The interface point $Z$ is sampled independently of
$(Q_j,U_{Q_j})$ from density $g_j$ on the supported
patch $\Gamma_j$.

For a fixed interface position $Z=z$, the estimator
in Eq.~\eqref{eq:hub-interface-continuation} satisfies
\begin{align}
\mathbb E[\widehat U_j(Z)\mid Z=z]
&=
\int_{\partial S_j}
r_j(q)
\frac{K_{S_j}(z,q)}{r_j(q)}
\mathbb E[U_{Q_j}\mid Q_j=q]
\,dA_q
\\
&=
\int_{\partial S_j}
K_{S_j}(z,q)u(q)\,dA_q
\\
&=
u(z).
\label{eq:app-hub-interface-unbiased}
\end{align}
The final equality is the Poisson representation of
the harmonic continuation inside $S_j$.
Thus every generated interface has the same pointwise
continuation semantics as an ordinary WoS state.

Now consider reuse from a query $x\in H$ through this
single group. Since $Z\sim g_j$ and
$g_j$ is supported on $\Gamma_j$,
\begin{align}
&\mathbb E\!\left[
\frac{K_H(x,Z)}{g_j(Z)}
\widehat U_j(Z)
\right]
\\
&\qquad=
\int_{\Gamma_j}
g_j(z)
\frac{K_H(x,z)}{g_j(z)}
\mathbb E[\widehat U_j(Z)\mid Z=z]
\,dA_z
\\
&\qquad=
\int_{\Gamma_j}
K_H(x,z)u(z)\,dA_z.
\label{eq:app-hub-patch}
\end{align}
Hence an individual group estimates exactly the
Poisson contribution of its supported patch.
If several such estimators were simply added,
overlapping patches would therefore be counted
multiple times.

We next consider the mixture construction.
Let
\[
\mathbb P(J=j)=\beta_j,
\qquad
Z\mid J=j\sim g_j,
\]
so that the joint density is
\begin{equation}
\mathbb Q\{J=j,Z\in dA_z\}
=
\beta_jg_j(z)\,dA_z,
\label{eq:app-hub-joint-mixture}
\end{equation}
and the marginal interface density is
\[
m(z)
=
\sum_{j=0}^{L}\beta_jg_j(z).
\]
Because the fallback group satisfies
$g_0(z)=\frac{1}{|\partial H|}$ and $\beta_0>0$,
we have $m(z)>0$ everywhere on $\partial H$.

For every retained group $j\ge1$,
Eq.~\eqref{eq:app-hub-interface-unbiased} gives
\[
\mathbb E[\widehat U_j(Z)\mid J=j,Z=z]
=
u(z).
\]
By construction, the fallback continuation also
satisfies
\[
\mathbb E[\widehat U_0(Z)\mid J=0,Z=z]
=
u(z).
\]
Therefore a single mixture sample contributes
\[
Y(x)
=
\frac{K_H(x,Z)}{m(Z)}
\widehat U_J(Z),
\]
whose expectation is
\begin{align}
\mathbb E[Y(x)]
&=
\sum_{j=0}^{L}
\int_{\partial H}
\beta_jg_j(z)
\frac{K_H(x,z)}{m(z)}
u(z)\,dA_z
\\
&=
\int_{\partial H}
\frac{
\sum_{j=0}^{L}\beta_jg_j(z)
}{
m(z)
}
K_H(x,z)u(z)\,dA_z
\\
&=
\int_{\partial H}
K_H(x,z)u(z)\,dA_z
\\
&=
u(x).
\label{eq:app-hub-mixture-unbiased}
\end{align}
Thus each mixture sample is an unbiased estimator of
the continuation at $x$.

Finally,
Eq.~\eqref{eq:hub-mixture-estimator} averages $N$
such interface contributions:
\[
\widehat u_H(x)
=
\frac1N\sum_{n=1}^{N}Y_n(x).
\]
Linearity of expectation yields
\[
\mathbb E[\widehat u_H(x)]
=
\frac1N\sum_{n=1}^{N}
\mathbb E[Y_n(x)]
=
u(x).
\]
Independence among the $Y_n$ is not required for this
expectation identity. In particular, multiple
interfaces may reference the same retained downstream
continuation and therefore remain correlated; such
correlation affects variance but not unbiasedness.

The graph weights
\[
w_{x\to Z_n}
=
\frac{K_H(x,Z_n)}{N\,m(Z_n)},
\qquad
w_{Z_n\to Q_j}
=
\frac{K_{S_j}(Z_n,Q_j)}{r_j(Q_j)}
\]
are precisely the factors appearing in the estimator
above, so the graph representation implements the same
unbiased mixture estimator.
\end{proof}

\section{Interface Construction by Path Refinement}
\label{app:fine-tuning-proof}

\subsection{Interface Construction}

We describe two ways to expose a retained continuation on a spherical
interface. Let $H=B(p_0,R_H)$ be a hub and
$S=B(q_0,R_S)$ the supporting ball of a retained edge
$q_0\to q_1$. Depending on their geometry, we construct an interface
$z\in\partial H$ together with a reference to the existing downstream
continuation.

\textbf{Intersecting spheres.}
Suppose $q_0\in H$ and $S$ partially overlaps $H$.
We condition the retained step on its exit $q_1$ and walk inside
$D=S\cap H$, whose signed distance is
\[
\operatorname{SDF}_D(x)
=\max\{\|x-q_0\|-R_S,\;\|x-p_0\|-R_H\}.
\]
For a local ball $B_k$ centered at the current position $x_k$ inside $D$,
the next position $x_{k+1}$ has density
\[
\frac{1}{|\partial B_k|}
\frac{K_S(x_{k+1},q_1)}{K_S(x_k,q_1)}.
\]
Let $X_t$ denote Brownian motion started at $q_0$, and let
$\tau_H$ and $\tau_S$ be its first exit times from $H$ and $S$.
Define $E=\{\tau_H<\tau_S\}$.
When $E$ occurs, the first crossing
$z=X_{\tau_H}\in\partial H$ is inserted into the retained edge with
\[
\omega_{q_0\to z}=\omega,
\qquad
\omega_{z\to q_1}=1,
\]
where $\omega$ is the original averaging coefficient.
If the conditioned segment reaches $\partial S$ first, the current edge is
left unchanged and refinement may continue on subsequent retained edges.
For an importance-sampled $q_1$, we apply the one-shot endpoint
correction in Appendix~\ref{app:refinement-endpoints}; accepted branches place
the compensating factor $M$ on the first edge.

\textbf{Containing spheres.}
If $\overline H\subset S$, no reconstructed crossing is required.
Sampling $z$ uniformly on $\partial H$, independently of the retained exit
$q_1\sim r_S$, defines $\widehat U_z$ and the original edge correction can be factorized as
\begin{equation}
\begin{aligned}
\widehat U_z&=\frac{K_S(z,q_1)}{r_S(q_1)}U_{q_1},\\
\omega_{q_0\to z}&=\frac{K_S(q_0,q_1)}{K_S(z,q_1)},
\qquad
\omega_{z\to q_1}=\frac{K_S(z,q_1)}{r_S(q_1)}.
\end{aligned}
\label{eq:contained-interface}
\end{equation} so their product equals the original
$\frac{K_S(q_0,q_1)}{r_S(q_1)}$.
A query $x\in H$ reuses this interface through
$\omega_{x\to z}=|\partial H|K_H(x,z)$.

\begin{proposition}[Validity of interface constructions]
\label{prop:fine-tuning-validity}
For an ordinary retained exit, suppose its downstream estimator remains
conditionally unbiased given the constructed interface. Both constructions
then produce valid continuation states:
$\mathbb E[U_Z\mid E,Z=z]
=\mathbb E[\widehat U_z\mid Z=z]=u(z)$.
For the containing construction, where $Z$ is uniform on $\partial H$,
$\mathbb E[|\partial H|K_H(x,Z)\widehat U_Z]=u(x)$ for every $x\in H$.
The intersecting construction instead follows its first-crossing law;
its interfaces do not in general have uniform density on $\partial H$.
The crossing-state identity also holds for importance-sampled retained exits
after the one-shot endpoint correction.
\end{proposition}

\begin{proof}[Proof sketch]
For intersecting spheres, the conditioned walk and strong Markov property
give $Q_1\mid(E,Z=z)\sim K_S(z,\cdot)$; for containing spheres, validity
follows from the nested Poisson representation. Importance-sampled exits
use the same one-shot endpoint correction as
Proposition~\ref{prop:refinement-validity}. The following subsections give
the detailed arguments for both geometries.
\end{proof}

Several interfaces may share the same retained continuation.

\subsection{Proof for Intersecting Spheres}

Let $S=B(q_0,R_S)$ and $H=B(p_0,R_H)$ with $q_0\in H$.
Let
\[
\tau_S=\inf\{t:X_t\notin S\},
\qquad
\tau_H=\inf\{t:X_t\notin H\},
\]
and define
\[
E=\{\tau_H<\tau_S\},
\qquad
Z=X_{\tau_H}\quad\text{on }E.
\]
The retained exit is
$Q_1=X_{\tau_S}$, whose ordinary exit density is
\[
\mathbb P\{Q_1\in dq_1\mid q_0\}
=
K_S(q_0,q_1)dA_{q_1}.
\]

Conditioned on $Q_1=q_1$, the Doob transform with
$h(x)=K_S(x,q_1)$ is used until the path reaches
$\partial(S\cap H)$. Hence the construction law $\mathbb Q$ satisfies
\begin{align}
&\mathbb Q
\{E,Z\in dz,Q_1\in dq_1\}\\
&\quad=
K_S(q_0,q_1)dA_{q_1}\,
\mathbb P
\{E,Z\in dz\mid Q_1=q_1\}\\
&\quad=
\mathbb P
\{E,Z\in dz,Q_1\in dq_1\}.
\end{align}
On $E$, the strong Markov property at $\tau_H$ gives
\[
\mathbb P
\{E,Z\in dz,Q_1\in dq_1\}
=
\mathbb P\{E,Z\in dz\}
K_S(z,q_1)dA_{q_1}.
\]
Therefore
\begin{equation}
\mathbb Q
\{Q_1\in dq_1\mid E,Z=z\}
=
K_S(z,q_1)dA_{q_1}.
\label{eq:app-interface-conditional-law}
\end{equation}

If the retained downstream estimator is conditionally unbiased,
\[
\mathbb E[U_{q_1}\mid q_1,E,Z=z]
=
u(q_1),
\]
then
\begin{align}
\mathbb E[U_Z\mid E,Z=z]
&=
\int_{\partial S}
K_S(z,q_1)u(q_1)\,dA_{q_1}\\
&=
u(z).
\end{align}
Thus the crossing point is a valid continuation state.

When $E$ occurs, replacing
$q_0\to q_1$ by
$q_0\to z\to q_1$ with coefficients
$\omega$ and $1$ leaves the realized downstream contribution of that
retained sample unchanged. If $E^c$ occurs, the current edge is left
unchanged.

If $q_1$ was generated from an importance density rather than the ordinary
Poisson law, we first apply the one-shot endpoint correction from
Appendix~\ref{app:refinement-endpoints}. Conditional on acceptance, the exit
again follows $K_S(q_0,\cdot)$ and the argument above is unchanged.
The compensating factor $M$ is placed on the first edge; a rejected attempt
contributes zero.

\subsection{Proof for Containing Spheres}

Assume $\overline H\subset S$.
Harmonicity gives the nested Poisson representation
\begin{equation}
u(x)
=
\int_{\partial H}
K_H(x,z)
\int_{\partial S}
K_S(z,b)u(b)\,dA_b\,dA_z,
\qquad x\in H .
\label{eq:app-nested-poisson}
\end{equation}

Let the retained exit $Q_1$ be sampled from area density $r_S$ and let
$U_{Q_1}$ be conditionally unbiased.
Independently sample $Z$ uniformly on $\partial H$ and define
\[
\widehat U_Z
=
\frac{
K_S(Z,Q_1)
}{
r_S(Q_1)
}
U_{Q_1}.
\]
Conditioning on $Z=z$,
\begin{align}
\mathbb E[\widehat U_Z\mid Z=z]
&=
\int_{\partial S}
r_S(q)
\frac{K_S(z,q)}{r_S(q)}
u(q)\,dA_q\\
&=
\int_{\partial S}
K_S(z,q)u(q)\,dA_q\\
&=
u(z).
\end{align}
Thus the auxiliary interface has the same pointwise continuation semantics
as an actual WoS state, although it is not a reconstructed crossing of the
old path.

For any $x\in H$, define
\[
\widehat U_H(x)
=
|\partial H|K_H(x,Z)\widehat U_Z .
\]
Since $Z$ is uniform on $\partial H$,
\begin{align}
\mathbb E[\widehat U_H(x)]
&=
\int_{\partial H}
\frac{dA_z}{|\partial H|}
|\partial H|K_H(x,z)u(z)\\
&=
u(x).
\end{align}

Finally, the edge factorization
\[
\omega_{q_0\to z}
=
\frac{K_S(q_0,q_1)}
     {K_S(z,q_1)},
\qquad
\omega_{z\to q_1}
=
\frac{K_S(z,q_1)}
     {r_S(q_1)}
\]
satisfies
\[
\omega_{q_0\to z}\omega_{z\to q_1}
=
\frac{K_S(q_0,q_1)}
     {r_S(q_1)},
\]
and therefore preserves the original sampled correction exactly.
Any pre-existing sample-averaging coefficient is retained on the first
edge.

Multiple interface points may share the same $U_{Q_1}$.

\section{Conditional Repair and Bidirectional Support}
\label{app:conditional-repair}

\subsection{Conditional Segments in the Updated Geometry}

Let $O$ be the inserted obstacle, let $S=B(c,R)$
support the retained transition $o\to p$ with $p\in\partial S$, and let
$D=S\setminus\overline O$ be its local repair
domain. The origin $o$ lies in $D$; it need not be the center $c$, as an
attached transition can inherit another sampling sphere. In three dimensions,
define the positive harmonic function
\[
h_p(x)=\frac{R^2-\|x-c\|^2}{\|x-p\|^3},\qquad x\in S.
\]
It is proportional to $K_S(x,p)$, so the proportionality factor cancels in
the conditioned transition. At a working position $x_k$, choose
\[
\rho_k=\kappa\min\{R-\|x_k-c\|,\operatorname{dist}(x_k,O)\},
\qquad 0<\kappa<1,
\]
and sample $y\in\partial B(x_k,\rho_k)$ with density
\[
\frac{1}{4\pi\rho_k^2}\frac{h_p(y)}{h_p(x_k)}.
\]
Use the exact sampler in Appendix~\ref{app:refinement-sampler} at this radius.

An obstacle-boundary tolerance terminates the segment at the closest obstacle
point. This test precedes reconnection at the original sphere. Reconnection
occurs when the remaining sphere clearance falls below a second tolerance,
and is disabled when $p\in\overline O$. Finite-tolerance reconnection omits
possible collisions on the remaining conditioned segment. Queries must remain
in the updated domain. A boundary-projection edge without a recorded sphere is
replaced by a fresh WoS continuation in the updated domain.

The working positions are temporary. The resolved edge retains its coefficient
and references either the old exit or a new boundary endpoint.

\subsection{A Local Correction from Shared Reverse Paths}

For a legal retained exit $p$, let $v$ denote the value supplied by its
continuation after repair. In the idealized continuous limit, the expected
conditional continuation in $D$ is
\begin{equation}
F_p(x)=v+\frac{w_p(x)}{h_p(x)},
\label{eq:repair-conditional-correction}
\end{equation}
where $w_p$ is harmonic in $D$, with boundary data
\[
\begin{aligned}
w_p(b)&=h_p(b)\bigl(g(b)-v\bigr),
&&b\in\partial O\cap S,\\
w_p(b)&=0,
&&b\in\partial S\setminus\overline O.
\end{aligned}
\]
The Doob transform makes $h_pF_p$ harmonic in $D$. Subtracting
$h_pv$ gives the stated harmonic correction: the object supplies $g$, while
the old-sphere exit supplies $v$. The reverse term therefore supplies the
inserted object's correction to the retained continuation.
If $p$ is swallowed, take $v=0$ and prohibit old-exit reconnection.

We generate a shared reverse-path library in the updated physical domain,
with independent surface starts on each boundary component. The local
correction uses the inserted-object paths restricted to their first exit
from $S$. An object-surface start outside $S$ contributes zero. Balls wholly
inside $S$ are retained; a ball crossing $\partial S$ is subdivided using
conditional refinement on its recorded exit, and the path is stopped on its
first local exit. Only the substep Green contributions are used for that
crossing ball. This restriction preserves shared global sampling while
providing the killed paths needed in $D$.

For the inserted-object stratum, let $N_R$ be its number of reverse paths,
including zero contributions. Path $j$ starts at a surface sample
$b_j\sim\pi_O$ displaced by $\delta_j$ along the normal into $D$.
The displacement is chosen small relative to the original-sphere clearance.
Let $\mathcal B_j$ be its restricted sequence of sampling balls. Define
\[
\begin{aligned}
\mathcal G_j(x)
&=\sum_{B\in\mathcal B_j}
G_B(x,c_B)\mathbf1_{x\in B},\\
\lambda_j(x)
&=\frac{h_p(b_j)\mathcal G_j(x)}
{h_p(x)\pi_O(b_j)\delta_j},
\end{aligned}
\]
where $c_B$ is the center of $B$, $\pi_O$ is an area density, and $G_B$ is
the Dirichlet Green function of the ball. Set $\lambda_j=0$ for paths whose
surface starts outside $S$. The finite-offset estimate is
\begin{equation}
\widehat F_p(x)
=v+\frac{1}{N_R}\sum_{j=1}^{N_R}
\lambda_j(x)\bigl(g(b_j)-v\bigr).
\label{eq:repair-reverse-correction}
\end{equation}
The inward displacement approximates a boundary normal derivative; physical
and sampling-sphere termination tolerances introduce additional approximations.
The zero correction on the old sphere is integrated exactly. Other physical
boundary components are already represented through the retained continuation.

\subsection{Connection and Graph Representation}

The forward segment uses the same conditional transitions with or without
reverse support. A connection is considered only within a prescribed distance
of the object, outside the reverse-offset layer and away from the original
sphere's termination shell. To prevent a reverse path from selecting its own
evaluation location, its contribution is settled only when the forward
position is covered by at least $K$ distinct other reverse paths, where $K$
is a positive integer. Multiple balls from one path count once. Each pending
path is evaluated at its first eligible position, including a zero evaluation
when that path does not cover the position. The normalization remains $N_R$.
The common forward segment continues until all contributions settle or reach
an ordinary object hit or old-exit reconnection. This leave-one-out rule keeps
each path's stopping decision independent of that path's reverse samples.

At a reverse connection, contribution $j$ supplies
$(1-\lambda_j)v+\lambda_jg(b_j)$, scaled by $\frac{1}{N_R}$; unresolved contributions
use the usual boundary hit or reconnection. Their combination is stored as
edges to the retained exit and the sampled physical boundary points, with
the original edge coefficient applied once. Coefficients may be signed;
$v$ remains a graph dependency, so subsequent boundary-value updates propagate
through the same DAG.

\end{document}